\documentclass[sumlimits,intlimits,namelimits,10pt]{amsart}

\usepackage{geometry}		
\usepackage{changepage}
\usepackage{fullpage}		
\usepackage{graphicx}
\usepackage{epstopdf}
\usepackage{xspace}
\usepackage[hang]{footmisc}
\usepackage{enumitem}
  \setlist{nolistsep}
\usepackage{mathtools}	
\usepackage{amssymb}
\usepackage{amsfonts}
\usepackage[mathscr]{eucal}
\usepackage{bbm}
\usepackage{bm}

\makeatletter
\def\subsection{\@startsection{subsection}{3}%
  \z@{1\linespacing\@plus.8\linespacing}{.5\linespacing}%
  {\centering\large\normalfont\sffamily}}
\makeatother

\makeatletter
\def\section{\@startsection{section}{3}%
  \z@{1\linespacing\@plus.8\linespacing}{.8\linespacing}%
  {\centering\large\normalfont\bf\sffamily}}
\makeatother

\theoremstyle{plain}
  \newtheorem{theorem}{Theorem}[section]
  \newtheorem{lemma}[theorem]{Lemma}
  \newtheorem{proposition}[theorem]{Proposition}

  \newtheorem*{corollary}{Corollary}
\theoremstyle{remark}
  \newtheorem*{remark}{Remark}
  \newtheorem*{remarks}{Remarks}
\theoremstyle{definition}
  \newtheorem{example}[theorem]{Example}
  \newtheorem{definition}[theorem]{Definition}

\newenvironment{multirem}{\setcounter{step}{0}}{}
  \newcounter{step}
  \newcommand{\rem}{\par\refstepcounter{step}~\thestep.\space\ignorespaces}
\DeclareMathOperator{\st}{\mid}			
	\DeclareMathOperator{\such}{\mid}	
\DeclareMathOperator{\f:}{:}			
\DeclareMathOperator{\tr}{tr}			
\DeclareMathOperator{\im}{im}		
\DeclareMathOperator{\Aff}{Aff}		
\DeclareMathOperator{\Inv}{Inv}		
\DeclareMathOperator{\Coeff}{Coeff}	
\DeclareMathOperator{\Lat}{Lat}		
\DeclareMathOperator{\GL}{GL}		
\DeclareMathOperator{\Lie}{Lie}		
\DeclareMathOperator{\spec}{spec}		
\DeclareMathOperator{\eig}{eig}		
\DeclareMathOperator{\vspan}{span}	
\DeclareMathOperator{\rank}{rank}		
\DeclareMathOperator{\diag}{diag}		
\newcommand{\CC}{\mathbb{C}}	
\newcommand{\RR}{\mathbb{R}}	
\newcommand{\ZZ}{\mathbb{Z}}	
\newcommand{\NN}{\mathbb{N}}	
\newcommand{\dg}{^\dagger}	
\newcommand{\cH}{\bm{\mathcal{H}}}	
\newcommand{\bp}{\bm{\sigma}}	
\newcommand{\np}[1]{\varphi[#1]}	
\newcommand{\A}{\mathscr{A}}	
\newcommand{\Gr}{\mathscr{G}}	
\newcommand{\Nbd}{\mathscr{N}}	
\newcommand{\Z}{\mathcal{Z}}	
\newcommand{\Mat}{\mathscr{M}}	
\newcommand{\Ham}{\mathscr{H}}	
\newcommand{\Burn}{\mathscr{B}}		
\newcommand{\cL}{\mathscr{L}}	
\newcommand{\Eig}{\mathscr{E}}	
\newcommand{\D}{\mathscr{D}}	
\newcommand{\Pb}{\mathcal{P}}	
\newcommand{\Sb}{\mathcal{S}}	
\newcommand{\Sample}{\mathscr{S}}	
\newcommand{\wo}{\backslash}	
\newcommand{\cl}[2][3]{{}\mkern#1mu\overline{\mkern-#1mu#2}}	
\newcommand{\trans}{^{\mathsf T}}
\newcommand{\calhat}[1]{\skew{3}\widehat{#1}}
\newcommand{\bra}[1]{\left\langle{#1}\right\vert}
\newcommand{\ket}[1]{\left\vert{#1}\right\rangle}
\newcommand{\expect}[1]{\langle{#1}\rangle}
\newcommand{\braket}[2]{\langle#1 \vert #2 \rangle}
\newcommand{\Unit}{\mathbbm{1}} 
\newcommand{\Pvec}{\ket{\nu_{\alpha}}}  

\def\mathclap#1{\text{\hbox to 0pt{\hss$\mathsurround=0pt#1$\hss}}}

\usepackage{xfrac}
\usepackage{color}
\usepackage{amsaddr}

\usepackage{caption}
\usepackage{float}
\usepackage{colortbl}
\usepackage{subfigure}
\usepackage{booktabs}
\usepackage{underbraces}	
\usepackage{placeins}
\newcommand{\ra}[1]{\renewcommand{\arraystretch}{#1}}
\newcounter{tabular}
\usepackage[norelsize,boxed,lined,linesnumbered]{algorithm2e}
\makeatletter
\renewcommand{\@algocf@capt@boxed}{above}	
\makeatother

\newcommand{\erf}[1]{{\color{blue} Eq.~\eqref{#1}}}

\newcommand{\etal}{\textit{et al.}}
\newcommand{\ie}{\emph{i.e.,}~}
\newcommand{\eg}{\emph{e.g.,}~}
\newcommand{\viz}{\emph{viz.,}~}

\newcommand{\SNL}{Department of Scalable \& Secure Systems Research (08961),
Sandia National Laboratories, Livermore, CA 94550, USA}

\title{On model reduction for quantum dynamics:\\ symmetries and invariant subspaces}
\author{Akshat Kumar, Mohan Sarovar}
\address{\SNL}
\email{akskuma@sandia.gov, mnsarov@sandia.gov}

\begin{document}
\begin{abstract}
Simulation of quantum dynamics is a grand challenge of computational physics. In this work we investigate methods for reducing the demands of such simulation by identifying reduced-order models for dynamics generated by parameterized quantum Hamiltonians. In particular, we first formulate an algebraic condition that certifies the existence of invariant subspaces for a model defined by a parameterized Hamiltonian and an initial state. Following this we develop and analyze two methods to explicitly construct a reduced-order model, if one exists. In addition to general results characterizing invariant subspaces of arbitrary finite dimensional Hamiltonians, by exploiting properties of the generalized Pauli group we develop practical tools to speed up simulation of dynamics generated by certain spin Hamiltonians. To illustrate the methods developed we apply them to several paradigmatic spin models.
\end{abstract}
\maketitle

\section{Introduction}
\label{sec:intro}
Exact simulation of quantum dynamics is notoriously difficult because of the exponential growth of a quantum system's state space with the number of independent degrees of freedom it possesses. At the same time, recent results \cite{Poulin:2011jh} suggest that an exponentially small subset of this formal state space is accessed through \textit{realistic} dynamics. Realistic dynamics are defined as dynamics generated by Hamiltonians with locality and energy constraints evolving for a sub-exponential (in system size) time from a small set of initial states. In light of these results, we consider the task of identifying this physically relevant subset of states in Hilbert space. This is an instance of \emph{model reduction} in systems theory \cite{Antoulas:2009tb, Schilders:2008uc}. The strongest form of model reduction is when the initial state of the system is evolved within a non-trivial invariant subspace of the exponentially large formal Hilbert space. In this case, there is no approximation error in restricting the dynamical model to this invariant subspace by projection. Dynamics according to this reduced order model can be significantly more efficient\footnote{Throughout this work we reserve the term ``efficient" to refer to methods that require fewer resources than a full order simulation of a quantum system. The term should not be interpreted in the formal complexity theoretic sense (as polynomial complexity in number of degrees of freedom). When such formal notions are necessary, we will explicitly specify the asymptotic scaling behavior.} to simulate if the dimension of the invariant subspace is much smaller than the dimension of the full Hilbert space for the system. We note that projection-based model reduction generally considers a wider class of linear projections, where the subspace projected onto is not necessarily an invariant subspace of the dynamics \cite{Antoulas:2009tb, Schilders:2008uc}, and in such cases the reduced order model is an \emph{approximation} of the full order model. However, in this work we only consider the strongest form of model reduction that projects onto an invariant subspace (containing the initial state) and thus incurs no error.

In the following we develop a formal characterization of the invariant subspace spanned by a given initial state and parameterized quantum Hamiltonian model. After establishing notation and definitions in \textsection  \ref{sec:setup}, we formulate an algebraic characterization of the existence of invariant subspaces for general parametric finite dimensional Hamiltonians in \textsection  \ref{sec:prelim} by applying classical results from the theory of representations of groups and semigroups \cite{curtis1999pioneers}. This characterization essentially provides a condition that certifies whether a parameterized Hamiltonian model yields an invariant subspace when acting on a given initial state. Then in \textsection  \ref{sec:construct} we present two techniques to explicitly construct the invariant subspace if one exists. 

In the second part of the paper, beginning with \textsection  \ref{sec:pauli}, we specialize to models of many-body spin dynamics and exploit properties of the generalized Pauli group to formulate stronger results that characterize invariant subspaces in this setting. We also develop specialized tools to explicitly construct the physically relevant invariant subspace for Pauli Hamiltonians in an efficient manner. These tools are then applied in \textsection  \ref{sec:eg} to several paradigmatic spin models.

From a physics perspective, the methods we develop can be viewed as allowing one to exploit (unitarily representable) symmetries in a many-body system without having to know these symmetries \emph{a priori}. That is, reduced subspaces are associated to model symmetries, and the methods directly compute the projected dynamics in these reduced subspaces. 

In terms of previous work on model reduction for quantum systems, Mabuchi \etal ~ have developed projection-based model reduction methods for dissipative, measured quantum systems that allow efficient estimates of quantum states based on continuous measurement outcomes \cite{Han.Mab-2005, Mabuchi:2008uy, Hopkins:2008tx, Nielsen:2009vs}. Also, Nurdin has explored balanced truncation for linear quantum stochastic models \cite{Nurdin:2013ux} with the focus of generating reduced order models that have a physical realization. In contrast to these works, our approach is specifically aimed at generating reduced order models for many-body systems that are governed by Hamiltonian dynamics (closed systems). In addition, while these previous works have developed methods for formulating reduced order models that reproduce input-output relationships, our focus will be on reproducing the quantum state vector generated by the dynamics. Our approach has similar aims as the recently developed time-dependent density matrix renormalization group (t-DMRG) method for many-body simulation \cite{White:2004fd,Daley:2004hk}. However, unlike t-DMRG, we do not begin with an ansatz for the physically relevant set of states and simulate their dynamics. Instead, we aim to identify this set of states from the dynamical model and reduce the model by projecting onto them. Finally, in contrast to all of the methods above our approach does not result in approximative reduced order models but exact ones since it relies on identifying invariant subspaces that host the dynamics.

\section{Setup and notation}
\label{sec:setup}
We use $\cH_d$ to denote a $d$-dimensional Hilbert space, and given any vector space $V$ we denote by $L(V)$ the space of linear operators $T:V\rightarrow V$. $\Mat_d(\CC)$ denotes the space of all complex $d\times d$ matrices and $\GL_d(\CC)$ the subspace of all \emph{invertible} $d\times d$ matrices.

Let the quantum system of interest be $d$-dimensional, meaning that a state of the system, $\ket{\psi}$, is a normalized vector in $\cH_d \cong  \mathbb{C}^d$. Then, we consider the following general class of Hamiltonian models that generate system dynamics on $\cH_d$:

\begin{definition}
A time-independent Hamiltonian model \emph{with $M$ free linear parameters} (or an $M$-\emph{parameter Hamiltonian model} for short) is a self-adjoint operator valued function $\Ham \f: \RR^M \rightarrow  \Mat_{d}(\CC)$ of the form:
\begin{equation} \label{hamdef}
	\Ham(\lambda) = H_0 + \sum_{k=1}^M \lambda_k H_k\ \ \ \ \	(M > 0)
\end{equation}
where the $\{H_k \}_{k=0}^M$ are fixed finite-dimensional self-adjoint operators on $\cH_d$. 
\end{definition}
The operator $H_0$ is often called the \emph{free Hamiltonian} and the parameters $\lambda = \{\lambda_k\}_{k=1}^M$ are real \emph{tuning} parameters (collectively, we also refer to $\lambda$ as the tuning parameter). In the many-body physics context, these tuning parameters often determine the \emph{phase} of the physical system. We restrict to time-independent parameters $\lambda_k$ in this work for simplicity; but all of the results can be generalized to the time-dependent case (that is particularly relevant to quantum control). We will often abuse notation and denote a Hamiltonian as well as the induced model by the same symbol, $\Ham$. Given a state, $\ket{\psi_0}$, of the system at time $t_0$, the state of the system at time $t$ is given by\footnote{Throughout this paper we use units where $\hbar=1$.}
\begin{equation} \label{eq:u_map}
	\ket{\psi_{\lambda}(t)} = U_\lambda(t,t_0) \ket{\psi_0} = \textrm{exp}\left\{ -i \Ham(\lambda) (t-t_0)\right\} \ket{\psi_0}.
\end{equation}

Without loss of generality, we can take $\Ham$ to be traceless since any traceful component would simply generate a global phase factor that is physically irrelevant. Moreover, since $\lambda_k$ take arbitrary real values, including zero, $\tr(H_k) = 0$ for all $0\leq k\leq M$.

We are concerned with finding an invariant subspace for the dynamics generated by this Hamiltonian that contains the initial state $\ket{\psi_0}$. Importantly, the invariant subspace should be independent of the tuning parameters $\lambda$ since we do not wish to construct such a subspace for every possible combination of parameters. We collect these objectives in the following definition.

\begin{definition}
Let $\Ham$ be a time-independent Hamiltonian model on $\cH_d$, with $M$ free parameters. If $\ket{\psi_0}$ is the initial state for the dynamical evolution governed by $\Ham$, then we combinedly denote this model by the notation, $\Ham_{\ket{\psi_0}}$. Further, we call a non-trivial proper subspace $V \subset  \cH_d$ that is invariant for all of the operators in the full space $\Ham(\RR^M)$ (\ie for all parameter values), and contains the initial state $\ket{\psi_0}$, a \emph{reduced subspace} for $\Ham_{\ket{\psi_0}}$.
\end{definition}

In the second part of this work we focus attention on many-body spin-$\frac{1}{2}$ (qubit) models. In this context, let $\Pb_1 := \{\sigma_k \}_{k=0}^3$ be the set of Pauli matrices with $\sigma_0 = I_2$. It is clear that this set forms an orthogonal basis for the Hilbert space of all two-dimensional linear operators (\viz the operators on the space for one qubit) under the Hilbert-Schmidt inner product. Thus, the set of permutations of $n$-fold tensor
products,
\[
	\Pb_n := \{\sigma_{\smash[b]{j_1}} \otimes  \cdots  \otimes  \sigma_{\smash[b]{j_n}} \st 0 \leq  j_1, \ldots , j_n \leq  3 \text{ and }
		\sigma_{\smash[b]{j_k}} \in  \Pb_1\}
\]
forms an orthogonal basis for the space of operators acting on the space of $n$-spins, $\cH_{2^n} \cong  (\CC^2)^{\otimes n}$. We denote the canonical generators of the Pauli group on $n$ spins by $\Pb_n^* := \Pb_n \wo \{I_{2^n} \}$.

\section{Preliminaries} 
\label{sec:prelim}

As a consistency check, we begin by showing that our definition of a reduced subspace for a given model $\Ham_{\ket{\psi_0}}$ actually preserves the dynamics generated by $\Ham$ at every value of the tuning parameters $\lambda \in  \RR^M$, with initial state $\ket{\psi_0}$. To this end, we define a necessary piece of notation:
\begin{definition}
	If $T \in  L(\cH_d)$, then the $T$-\emph{cyclic subspace generated by} $\ket{\phi} \in  \cH_d$ is the subspace of
	$\cH_d$ given by $\Z(T, \ket{\phi}) := \vspan\{T^k \ket{\phi} \st k \geq  0 \}$.
\end{definition}

Since the dynamics of $\Ham$ are governed by the evolution equation \erf{eq:u_map}, we see by the convergent expansion,
\[
	U_\lambda(t,t_0) = \sum_{k=0}^\infty {(-i)^k \over k!} (t-t_0)^k \Ham(\lambda)^k ,
\]
that for each fixed $\lambda \in  \RR^M$ and every $t \in  \RR$, $\ket{\psi_\lambda(t)} \in  \Z(\Ham(\lambda), \ket{\psi_0})$. Now, if $V$ is a reduced subspace for $\Ham_{\ket{\psi_0}}$, then we have by definition that $\Ham(\lambda)V \subseteq  V$ and $\ket{\psi_0} \in  V$, for every $\lambda \in  \RR^M$. Thus clearly, given any $\lambda \in  \RR^M$, $\Ham(\lambda)^k \ket{\psi_0} \in  V$ for every $k \geq  0$, whence $\Z(\Ham(\lambda), \ket{\psi_0}) \subseteq  V$. This completes the proof of the following:

\begin{lemma} \label{prelim:check_lem}
	Given a Hamiltonian model $\Ham_{\ket{\psi_0}}$, every reduced subspace encompasses every state $U_\lambda(t,t_0)~\ket{\psi_0}$, for all $t,t_0 \in  \RR$, and all $\lambda \in  \RR^M$.
\end{lemma}

With the above definitions in place we can be more specific about how to construct reduced order dynamical equations for quantum Hamiltonian models possessing a reduced subspace. Let the columns of a $d\times r$ matrix $\Phi$ be a basis for an $r$-dimensional reduced subspace for $\Ham_{\ket{\psi_0}}$, with $\Phi\dg \Phi = I_r$. The dynamics generated by the Schr\"odinger equation for the original model,
\[
	\frac{{\textrm d}}{{\textrm d}t}\ket{\psi_\lambda(t)} = -i \Ham(\lambda) \ket{\psi_{\lambda}(t)}, \hspace{1cm} \ket{\psi_{\lambda}(t_0)}=\ket{\psi_0},
\] 
is equivalent to the dynamics generated by the projected model
\begin{equation}
\label{eq:red_order_model}
	\frac{{\textrm d}}{{\textrm d}t}\ket{\upsilon_\lambda(t)} = -i \calhat{\Ham}(\lambda) \ket{\upsilon_\lambda(t)}, \hspace{1cm} \ket{\upsilon_\lambda(t_0)}=\Phi\dg \ket{\psi_0},
\end{equation}
where the projected, reduced-order Hamiltonian is the $r\times r$ matrix
\begin{equation}
\label{eq:red_order_ham}
\calhat{\Ham}(\lambda) = \Phi \dg \Ham(\lambda) \Phi,
\end{equation}
for all $\lambda$. The $r\times 1$ vector, $\ket{\upsilon_\lambda(t)}$, is a compressed, faithful representation of the quantum state; the full representation of the state can be recovered at any time since $\ket{\psi_\lambda(t)} = \Phi \ket{\upsilon_\lambda(t)}$.

\subsection{Reduction to Common Invariant Subspaces for Finitely Many Operators}
We now show that we can readily reduce the problem of finding reduced order models for the dynamics generated by $\Ham_{\ket{\psi_0}}$ to finding an invariant subspace (containing $\ket{\psi_0}$) common to a finite collection of related operators, which we denote by  $\Coeff(\Ham)$. By basic properties of invariant subspaces, the problem is then cast as that of understanding the algebra of operators generated by $\Coeff(\Ham)$. In this way, we convert the problem into an algebraic one, which will be addressed in later sections.

By considering $\Ham(\lambda=\beta_k)$, for each canonical basis vector $\beta_k$ for $\RR^M$, along with $\Ham(0)$, and taking their linear combinations, we see that a subspace is invariant for $\Ham$ if and only if it is invariant for the collection $\{H_0 + H_k \}_{k=1}^M \cup \{H_0\}$ (which is one possible choice for $\Coeff(\Ham)$). In fact, by exploiting that $\Ham$, as given by \erf{hamdef}, is an affine map $\RR^M \rightarrow  L(\cH_d)$ as a function of the tuning parameter $\lambda$, we can strengthen this consideration to point out that we can choose any $M+1$ affinely independent points in $\RR^M$ and the existence of an invariance for the full model is equivalent to the existence of a common invariant subspace for $\Ham$ at these points \footnote{For definitions of affine spaces and associated concepts (affine independence, affine basis, affine rank) see Ref. \cite[Ch. 2]{Gallier:2011vu}.}. A small generalization of this is the content of the following:

\begin{lemma} \label{prelim:lem1}
	Let $\mathcal{A}: \CC^M \rightarrow  L(\cH_d)$ be an affine map with full affine rank, for any $M \in  \NN$ and
	denote by $\mathcal{A}_L$ its linear part. Then, for $V \subset  \cH_d$ a non-trivial proper subspace, the
	following are equivalent:
	\begin{enumerate}[label=(\roman{*})]
	\item	If $\beta$ is an affine basis for $\CC^M$ and $\beta_* \subset  \beta$ is a linear basis, $V$ is invariant for the set
			$\mathcal{A}_L(\beta_*) \cup  \mathcal{A}(\beta \wo \beta_*)$ (and thus also for $\mathcal{A}(\beta)$);
				\label{prelim:lem1.i}
	\item	There exists an affine basis $\beta$ for $\CC^M$ such that $V$ is invariant for the set
			$\mathcal{A}(\beta)$;	\label{prelim:lem1.ii}
	\item	$V$ is invariant for the affine space $\im(\mathcal{A}) := \{\mathcal{A}(\gamma) \st \gamma \in  \CC^M \}$.
				\label{prelim:lem1.iii}
	\end{enumerate}
\end{lemma}

\begin{proof}
	These results follow trivially from the \emph{affineness} of $\mathcal{A}$ and the preservation of
	the invariance of a subspace under linear combinations of maps for which that subspace is invariant.
	Note also that whenever $\beta \subset  \CC^M$, the affine hull $\Aff(\beta)$ has dimension at most the linear hull
	of $\beta$. Therefore, if $\beta$ is an affine basis, then it must contain a linear basis $\beta_*$ and
	$\beta \wo \beta_* = \{\gamma \}$. Then, choosing $\beta$ to be the canonical orthonormal basis for $\CC^M$ along with
	the zero vector, we immediately get \ref{prelim:lem1.i} $\implies$ \ref{prelim:lem1.ii}.
	
	For \ref{prelim:lem1.ii} $\implies$ \ref{prelim:lem1.iii}, note that since $\beta$ is an affine basis for $\CC^M$,
	$\Aff(\mathcal{A}(\beta)) = \im(\mathcal{A})$. And since affine combinations are just linear combinations
	with a restriction on the coefficients, we have by linearity that $V$ is invariant for $\im(\mathcal{A})$.
	
	Now we need only to see that \ref{prelim:lem1.iii} $\implies$ \ref{prelim:lem1.i}. Let $\beta := \beta_* \cup  \{\gamma \}$ be
	an affine basis for $\CC^M$ with $\beta_*$ a linear basis. By the assumption, $V$ must be invariant for
	$\mathcal{A}(\beta) = \mathcal{A}(\beta_*) \cup  \{\mathcal{A}(\gamma) \}$. Thus, by linearity, it must also be
	invariant for $\mathcal{A}_L(\beta_*) = \mathcal{A}(\beta_*) - \mathcal{A}(\gamma)$.
\end{proof}

\begin{remarks}\begin{multirem}
\rem	
This is just an elaboration on the observation that invariant subspaces for a collection of operators are
	preserved under affine transformations of the scalar coefficient variables describing the affine
	combinations among the operators. 

\rem	
In the sense of inhomogeneous systems, it is really the classical Rouch\a'e-Capelli theorem that is at
	work here, allowing us to rewrite the constituent operators of $\mathcal{A}$ as linear combinations
	of $\mathcal{A}(\gamma)$ evaluated at $M+1$ affinely independent points. The trend continues: if we take
	$\mathcal{A}$ to be linear, then it is clear that invariance for $\mathcal{A}$ by anything less than
	a basis for $\CC^M$ immediately implies a linear dependence among its constituent operators (and
	conversely). A much stronger version of this is the content of a theorem proven by Burnside, which
	will be the topic of discussion in \textsection \ref{sec:cert_general}.
\end{multirem}\end{remarks}

We record the application of Lemma \ref{prelim:lem1} to reducing our original problem, in the following:

\begin{proposition} \label{prelim:prop1}
	An $M$-parameter Hamiltonian model acting on $\cH_d$,
	\[
		\Ham(\lambda) = H_0 + \sum_{k=1}^M \lambda_k H_k
	\]
	has a non-trivial proper invariant subspace $V \subset  \cH_d$ (independent of $\lambda \in  \RR^M$) if and only if $V$ is invariant for a set of operators $\Coeff(\Ham)$ given by the image of an affine basis $\beta$ of $\RR^M$ under $\Ham$, or equivalently, $\Ham(\beta_*) \cup  \Ham(\beta \wo \beta_*)$, for $\beta_*$ a linear basis in $\beta$.
	\qed
\end{proposition}
\begin{remark}
	Letting $\beta := \{\beta_k\}_{k=1}^M$ be the canonical orthonormal basis for $\RR^M$ and $\Ham_L$ the
	linear part of $\Ham$, we see that $\Ham_L(\beta) \cup  \Ham(0) = \{H_k \}_{k=0}^M$.  Thus, by the above,
	we can take $\Coeff(\Ham)$ to be this set of operators, which is a rather natural choice. However, in the
	sequel, we keep to the understanding that $\Coeff(\Ham)$ is any collection of operators governed by
	Proposition \ref{prelim:prop1}.
\end{remark}

Since the invariance of subspaces is preserved not only under taking linear combinations, but also under compositions of operators, the question of finding reduced subspaces for a model $\Ham_{\ket{\psi_0}}$ really involves the full algebraic structure of the associated operator algebra. This is convenient since characterizing substructures of the operator algebra $L(\cH_d)$ is easier than characterizing subspaces of the Hilbert space $\cH_d$. This is the topic of the next subsection, but first we list some definitions that will be useful in what follows.

\begin{definition}	\label{def:algebra}
Let $S = \{ A_1, A_2, ..., A_K\}$ be a set of operators with each $A_i \in  L(\cH_d)$. The operator algebra $\A(S)$, which is a subalgebra of $L(\cH_d)$, is defined as
\[
	\A(S) = \{ p(A_1, A_2, ..., A_K) | ~ p \text{ is a polynomial in $K$ non-commuting variables} \}
\]
\end{definition}

For a Hamiltonian model $\Ham$ acting on $\cH_d$, we will slightly abuse notation for simplicity and denote $\A(\Ham)$ as the algebra generated by $\Coeff(\Ham)$, since $\A(\Ham) = \A(\Coeff(\Ham))$. $\A(\Ham)$ will play a central role in algebraic certificates for the presence of a reduced subspace for $\Ham_{\ket{\psi_0}}$.

Next, it will be useful (mostly for notational convenience) to define the notion of a \emph{lattice} of subspaces of a given vector space (see \cite[\textsection~1.8]{israel2006invariant}).
\begin{definition} \label{def:lattice}
Given a vector space $V$, a set $\mathcal{L}$ of subspaces of $V$ is called a \emph{lattice} if $\{0\}$ and $V$ belong to $\mathcal{L}$ and it is closed under taking intersections and direct sums. Clearly the elements of $\mathcal{L}$ are (partially) ordered by the subspace inclusion relation, which we will denote by $\leq $; collectively, we denote a lattice along with its order relation as the tuple, $(\mathcal{L}, \leq )$.
\end{definition}

\noindent As an example, given any vector space $V$, the family of all subspaces of $V$ trivially forms a lattice, $(\Lat(V), \leq )$. Thus, whenever a subspace $W$ of $V$ contains another subspace $W'$, we can express this fact by the notation, $W' \leq  W$ and we attain a \emph{chain} of inclusions: $\{0\} \leq  W' \leq  W \leq  V$. We also define the following notation that allows us to consider just the non-trivial elements of a subspace lattice:

\begin{definition}
If $V$ is a finite-dimensional vector space and $\mathcal{L}$ is a lattice on $V$, then we denote by $\mathcal{L}^\circ $ its subset of all but the extremal (minimal and maximal) elements, $\{0\}$ and $V$.
\end{definition}

For a set $S$ of $d$-dimensional operators, we will collect its family of invariant subspaces in $\Inv(S) := \{V \in  \Lat(\cH_d) \st \forall A\in S, AV \subseteq  V \}$, which is obviously closed under taking direct sums and intersections and thus forms a lattice $(\Inv(S), \leq )$; clearly this is contained in $\Lat(\cH_d)$.

Finally, following the language of representation theory, we make the definitions: 
\begin{definition}
A collection of operators $S$ is \emph{irreducible} if $\Inv(S)^\circ  = \varnothing $. Similarly, a subalgebra $\A$ is \emph{irreducible} if $\Inv(\A)^\circ  = \varnothing $.
\end{definition}
\noindent Irreducibility means that the only invariant subspaces common to all operators in $S$ (or the algebra $\A$) are the trivial subspaces, $\{0\}$ and $\cH_d$.

\subsection{A certificate for the existence of invariant subspaces}
\label{sec:cert_general}
A classical result in the theory of matrix algebras, stemming from the representation theory of groups, shows that we can determine the existence of an invariant subspace for the Hamiltonian model $\Ham$ by examining the \emph{dimension} of $\A(\Ham)$ as a linear subspace of $L(\cH_d)$. Indeed, the general theorem we use is the following:

\begin{theorem}[\cite{Lomonosov200445}, \cite{jacobson1953lectures}] \label{red:burnside}
	A subalgebra $\A \leq  \Mat_d(\CC)$ is irreducible if and only if $\A = \Mat_d(\CC)$.
\end{theorem}

\begin{remark}
	See Appendix \ref{app:burnside_proof} for a derivation of this theorem from the more classical version
	that was proved originally by Burnside and belongs to group representation theory.
\end{remark}

This immediately results in the following certificate for the existence of a non-trivial proper invariant subspace for a Hamiltonian model:

\begin{proposition} \label{red:burnprop}
	A Hamiltonian model $\Ham(\lambda)$ acting on $\cH_d$ keeps invariant a non-trivial proper subspace of $\cH_d$ if and only if the subalgebra $\A(\Ham)$ generated by $\Coeff (\Ham)$ is a proper subalgebra of $L(\cH_d)$ -- viz., $\dim(\A(\Ham)) < \dim(L(\cH_d))$.
	\qed
\end{proposition}

Application of this certificate requires determining $\dim(\A(\Ham))$. There are various methods to calculate this quantity given $\Coeff (\Ham)$, and we take the approach of constructing a basis for $\A(\Ham)$, which we call the \emph{Burnside basis}, and denote by $\mathscr{B}(\Ham)$. 
This basis is a maximal linearly independent subset (over $\CC$) of all monomials in $\A(\Ham)$ generated by taking products of the operators in $\Coeff(\Ham)$. An explicit algorithm for generating $\Burn(\Ham)$ is detailed in Appendix \ref{app:burnside_alg}; it proceeds by multiplying elements of $\Coeff(\Ham)$ together, and then the resulting matrices with elements of $\Coeff(\Ham)$ again, and so on. Subsets of these matrix products are added to $\Burn(\Ham)$ according to some decision procedure to determine linear independence. The algorithm terminates after a finite number of steps since the algebra is finite-dimensional, and $\dim(\A(\Ham)) = |\mathscr{B}(\Ham)|$.

This certificate for the existence of a non-trivial proper invariant subspace is computationally straightforward and has the benefit that it yields information (the Burnside basis) that enables construction of a \emph{reduced} subspace for $\Ham_{\ket{\psi_0}}$, a task that is addressed in \textsection  \ref{sec:burnside_alg}. We postpone a discussion of other possible certificates to \textsection  \ref{sec:disc}.

\subsection{Structure of $\A(\Ham)$}
\label{sec:struct_A}
The study of invariant subspaces for a Hamiltonian model $\Ham(\lambda)$ inevitably requires consideration of the full operator algebra $\A(\Ham) = \A(\Coeff(\Ham))$. Since all operators in $\Coeff(\Ham)$ are self-adjoint, the algebra $\A(\Ham)$ is also \emph{self-adjoint}, that is, $A^\dagger \in  \A(\Ham)$ whenever $A \in  \A(\Ham)$. Moreover, for any operator $T \in  L(\cH_d)$, the spaces $\Inv(T)$ and $\Inv(T^\dagger)$ are in one-to-one correspondence by taking orthogonal complements, which means that $\Inv(\A(\Ham))$ is closed under orthogonal complementation. Thus, for any chain $\{0 \} = V_0 < \cdots  < V_K = \cH_d$ in
$\Inv(\A(\Ham))$, we have the orthogonal decomposition \cite[\textsection~11.2]{israel2006invariant}:
\begin{equation} \label{red:decomp}
	\cH_d = \bigoplus_{k=1}^K \mathcal{E}_k,
\end{equation}
where $\mathcal{E}_k \in  \Inv(\A(\Ham))$ is the orthogonal complement of $V_{k-1}$ in $V_k$, for each $1\leq k\leq K$.
Therefore, $\A(\Ham)$ can be simultaneously transformed to a block-diagonal form
\begin{equation} \label{red:blkdiag}
	UAU^\dagger =
	\left( \begin{array}{ccc}
	A_1	& 	& \makebox(0,0){\text{\large 0}} \\
		& \ddots \\
	\makebox(0,0){\text{\large 0}} & & A_K
	\end{array} \right),
\end{equation}
for each $A \in  \A(\Ham)$, with $U$ the unitary change of basis transformation from the canonical orthonormal basis for $\cH_d$ to an orthonormal basis for the decomposition in \erf{red:decomp}. It is well-known (see \cite[Chap. 4]{jacobson1953lectures}) that the converse of this discussion also holds, due to which we are led to the following simple characterization:
\begin{proposition} \label{red:prop1}
	Let $\Ham(\lambda)$ be a Hamiltonian model acting on $\cH_d$. Then, $\Inv(\A(\Ham))^\circ  \neq  \varnothing $ if and only if $\Coeff(\Ham)$ is non-trivially simultaneously block-diagonalizable by a unitary transformation $U \in  L(\cH_d)$.
\end{proposition}

\begin{remark}
	In fact, one need not necessarily find a \emph{unitary} block-diagonalizing matrix: a result of Laffey \cite[Lemma 1]{Laffey1977189} shows that the $\Coeff(\Ham)$ is non-trivially simultaneously block-diagonalizable by a unitary if and only if $\Coeff(\Ham)$ can be non-trivially simultaneously block-diagonalized by some non-singular matrix. See \cite[\textsection~2]{Shapiro1979129} for more discussion.
\end{remark}

If we denote by $\A_k(\Ham)$ the projection of $\A(\Ham)$ onto $\mathcal{E}_k$ ($1\leq k\leq K$), then each $\A_k(\Ham)$ ($1\leq k\leq K$) forms an algebra consisting of the $A_k$ that appear in the block-diagonalization above, corresponding to each $A \in  \A(\Ham)$. Since it may happen that for every $A \in \A(\Ham)$, there is some $1 \leq k \leq K$ such that $A_k$ in \erf{red:blkdiag} is repeated $j$ times, we can collect all repetitions consecutively and denote this by the following notation:
\[
	\A_k(\Ham) \otimes \Unit_j := \{ \underbrace{A_k \oplus \cdots \oplus A_k}_\text{$j$ times} \st A_k \in \A_k(\Ham) \}.
\]
This allows us to write a decomposition analogous to \erf{red:decomp} for the operator algebra:
\[
	U \A(\Ham) U^\dagger = \{A_1 \oplus  \cdots  \oplus  A_K \such A_k \in  \A_k(\Ham), 1\leq k\leq K \} = \bigoplus_{k=1}^{K'} \A_k(\Ham) \otimes \Unit_{j_k}.
\]
If we further continue the decomposition in \erf{red:decomp} by taking the corresponding chain in $\Inv(\A(\Ham))$ to be \emph{maximal}, meaning that the $\mathcal{E}_k$ do not contain any invariant subspaces, then each $\A_k(\Ham)$ must be irreducible and so by Theorem \ref{red:burnside}, $\A_k(\Ham) = \Mat_{q_k}(\CC)$, with $q_k = \dim(\mathcal{E}_k)$. Accordingly, it is well-known (see \cite[\textsection ~ 5.1]{farenick2001algebras}) that there is a unitary $U$ such that the \emph{full} decomposition of the algebra is
\begin{equation} \label{red:fulldecomp}
	U \A(\Ham) U^\dagger = \bigoplus_{k=1}^m \Mat_{q_k} \otimes \Unit_{j_k},
\end{equation}
where $m$ is the dimension of the centre of $\A(\Ham)$.

\vspace{2mm}

\begin{remarks}\begin{multirem}
\rem		In terms of Proposition \ref{red:prop1}, the existence of the invariant subspaces that we will exploit below is equivalent to the existence of a unitary transformation that block diagonalizes the Hamiltonian model $\Ham(\lambda)$ for any value of $\lambda$. Then by Wigner's theorem \cite{Bargmann:1964eo} what we are are exploiting from a physics perspective are unitary symmetries of the system.
\rem		Recent methods for performing simultaneous block diagonalization of matrix algebras \cite{Murota:2010wu} could be used to construct reduced subspaces of Hamiltonian models based on Proposition \ref{red:prop1}. These methods rely on eigenspace computations and therefore will generally be sensitive to numerical uncertainty.  In contrast, the algebraic method we develop in the following involves only rational computations, aside from the inherent approximations that may be necessary for the \emph{a priori} specification of the system (\eg specification of the initial state).
\end{multirem}\end{remarks}


\section{Constructing reduced subspaces for a Hamiltonian model}
\label{sec:construct}
In \textsection  \ref{sec:cert_general} we derived an algebraic condition for the existence of a non-trivial invariant subspace for the Hamiltonian model $\Ham(\lambda)$. Evaluating this algebraic condition (which is essentially a comparison between the dimension of the subalgebra generated by $\Coeff (\Ham)$ and the dimension of the space of all operators on the state space) provides a certificate for whether a non-trivial invariant subspace will exist for the model. In this section we extend this result to explicitly construct a reduced subspace for the Hamiltonian model when it evolves a particular initial state. This will allow us to exploit the reduced subspace to construct a reduced-order model for the dynamics, as in \erf{eq:red_order_model}. In the following we present two different techniques to explicitly construct the reduced subspace of interest. 

\begin{remark}
	We note for completeness that after the construction of $\A(\Ham)$ one may wish to check if $\ket{\psi_0}\bra{\psi_0} \in \A(\Ham)$. If so, every invariant subspace of the subalgebra is also a reduced subspace or the orthogonal complement of a reduced subspace. In this case, techniques to find invariant subspaces of the subalgebra $\A(\Ham)$ can be utilized to identify the minimal reduced subspace. In the following we do not make this assumption, but rather formulate general methods for finding reduced subspaces.
\end{remark}

Before describing the methods, it will be useful to define the orbit of a state under a subalgebra:
\begin{definition}
\label{def:orbit}
Given a state $\ket{\phi} \in  \cH_d$ and a subalgebra $\A \leq  L(\cH_d)$, the \emph{orbit} of $\ket{\phi}$ under $\A$ is denoted and defined as $\A\cdot  \ket{\phi}= \{A \ket{\phi} \st A \in  \A \}$.
\end{definition}
\noindent Note that $\A\cdot  \ket{\phi}$ is a subspace of every invariant subspace for $\A$ that contains $\ket{\phi}$. In terms of our Hamiltonian model $\Ham_{\ket{\psi_0}}$, this means that the orbit of the state $\ket{\psi_0}$ under $\A(\Ham)$ is the \emph{minimal} reduced subspace for the model. Hence, a reduction for $\Ham$ exists if and only if $\dim(\A(\Ham)\cdot \ket{\psi_0}) < \dim(\cH_d) = d$.

\subsection{Construction via the Burnside basis}
\label{sec:burnside_alg}
Any linear basis, $B$, for $\A(\Ham)$ can be used to generate a spanning set for the minimal reduced subspace $\A(\Ham)\cdot\ket{\psi_0}$.
Therefore we can use the Burnside basis, $B=\mathscr{B}(\Ham)$, to construct a basis for $\A(\Ham)\cdot\ket{\psi_0}$ by taking a maximal linearly independent subset of $\{\ket{\beta_j} = B_j \ket{\psi_0} ~|~ B_j \in \mathscr{B}(\Ham)\}$ and orthonormalizing it. The resulting vectors form the columns of the model reduction matrix, $\Phi$. Note that $I_d \in \mathscr{B}(\Ham)$ by construction, and therefore the initial state is guaranteed to be in the reduced subspace. Also, $\rank(\Phi) \leq |\mathscr{B}(\Ham)|$ since although the elements of $\mathscr{B}(\Ham)$ are all linearly independent, the $\ket{\beta_j}$ are not necessarily all mutually orthogonal.

\vspace{2mm}

\begin{multirem}\begin{remarks}
\rem	Note that $\ket{\psi_0} \in  V < \cH_d$ if and only if $V \in  \Inv(\ket{\psi_0}\bra{\psi_0})^\circ $. Thus, any procedure for certifying the existence of invariant subspaces for $\Ham$ by operating on $\Coeff(\Ham)$ can be modified to certify the existence of \emph{reduced} subspaces for $\Ham_{\ket{\psi_0}}$. Explicitly, a reduced subspace for $\Ham_{\ket{\psi_0}}$ exists if and only if the dimension of the subalgebra $\A(\Coeff(\Ham)\cup \{\ket{\psi_0}\bra{\psi_0} \})$ is less than the dimension of $L(\mathcal{H}_d)$. Further, a basis for $\A(\Coeff(\Ham)\cup \{\ket{\psi_0}\bra{\psi_0} \})$ gives the relevant reduced subspace. However, it is often computationally more efficient to generate the orbit of $\ket{\psi_0}$ under $\A(\Ham)$ since a simple algebraic characterization of the state $\ket{\psi_0}$ may not exist. We shall see an example of this when we consider Pauli Hamiltonians in \textsection  \ref{sec:pauli}.

\rem
We note that the certificate provided by Proposition \ref{red:burnprop} bears some resemblance to the Lie algebraic controllability condition derived in the context of unitary quantum control \cite{Ram.Sal.etal-1995}; however, the physical context for the two conditions are different. Although it is true that for the Lie algebra $\Lie(\Ham)$ generated by $\Coeff(\Ham)$, $\Inv(\Lie(\Ham)) = \Inv(\A(\Ham))$, the Lie algebra is generally not able to generate invariant subspaces by taking orbits of vectors. This is simply due to the fact that $\Lie(\Ham)$ is a vector subspace of $\A(\Ham)$.\footnote{$\A(S)$ is sometimes referred to as the \emph{interaction algebra} in the literature surrounding decoherence-free subspaces, where $S$ is a set operators defined by the interaction Hamiltonian between system and environment. In particular, the authors of \cite{Kempe:2001tr} note that the Lie algebra over $S$ is not equal to the interaction algebra in general.}
Since $\A(\Ham)\cdot \ket{\psi_0}$ gives the minimal reduced subspace for $\Ham_{\ket{\psi_0}}$, a strict inequality in the containment $\Lie(\Ham)\cdot \ket{\psi_0} \leq  \A(\Ham)\cdot \ket{\psi_0}$ would imply that the Lie algebra does not generate a reduced subspace at all. An explicit example of this is provided by the following $4$-spin Hamiltonian model 
\[
		\Ham(\lambda) = \sum_{j=1}^4 \sigma_z^{(j)} + \lambda_x \sum_{j=1}^4 \sigma_x^{(j)} + \lambda_y \sum_{j=1}^4 \sigma_y^{(j)},
\]
where $j$ denotes the spin on which the Pauli matrix acts non-trivially; \ie $\sigma_\alpha^{(1)} := \sigma_\alpha \otimes I_2 \otimes I_2 \otimes I_2$, and so on. Also, assume the initial state is an eigenstate of $\sum_{j=1}^4 \sigma_z^{(j)}$. Then $\dim(\Lie(\Ham))=3$ since the permutation symmetry of $\Ham$ makes the action isomorphic to that of $su(2)$. Furthermore, the subspace $\Lie(\Ham)\cdot \ket{\psi_0}$ is $2$-dimensional. In contrast, $\dim(\A(\Ham))=35$ and the minimal reduced subspace, $\A(\Ham)\cdot \ket{\psi_0}$, is $5$-dimensional. 
\end{remarks}\end{multirem}

\subsection{Construction by sampled time evolution}
\label{sec:sample}
In this subsection we explore an alternative method for constructing the physically relevant reduced subspace by using samples of the quantum evolution, $\ket{\psi_{\lambda}(t)} := e^{-i t \Ham(\lambda)} \ket{\psi_0}$, at some points in time, $t_1, \ldots , t_m$ ($m \geq 1$) and for some fixed value(s) of the tuning parameter $\lambda$. The subspace spanned by the states (or \emph{snapshots}) $\{\ket{\psi_\lambda(t_1)}, \ldots , \ket{\psi_\lambda(t_m)}\}$ readily contains some of the dynamics governed by the Hamiltonian $\Ham(\lambda)$. In this section, we will improve upon this observation by showing conditions under which it is possible for the finite collection of snapshots to span the minimal reduced subspace, $\A(\Ham)\cdot \ket{\psi_0}$, for the full model $\Ham_{\ket{\psi_0}}$. We will furthermore see that when these conditions are met, taking enough ($\leq  d$) samples at random from any time interval, according to any continuous distribution on the interval, will \emph{almost always} yield a spanning set for $\A(\Ham)\cdot \ket{\psi_0}$. In this case the columns of the model reduction matrix $\Phi$ are formed from a maximal linearly independent subset of $\{\ket{\psi_\lambda(t_1)}, \ldots , \ket{\psi_\lambda(t_m)}\}$.

The motivation for considering this approach to constructing the reduced subspace is that in some contexts it may be practical to directly simulate the full order model for a short time at a few values of the tuning parameter. The question is whether the results of such large scale simulations can be used to identify the reduced subspace for $\Ham_{\ket{\psi_0}}$ and thereby construct a reduced order model. 

We shall keep fixed the value of the tuning parameters at $\lambda \in  \RR^M$ for the model $\Ham_{\ket{\psi_0}}$. Let $\Psi_\lambda(I) := \vspan \{\ket{\psi_\lambda(t)} \st t \in  I \}$, for $I \subseteq  \RR$ any interval. The span of our samples $\{\ket{\psi_\lambda(t_1)}, \ldots , \ket{\psi_\lambda(t_m)}\}$ for $t_i \in I$ will be a subspace of $\Psi_\lambda(I)$, and we will address below how the dimensions of $\vspan \{\ket{\psi_\lambda(t_1)}, \ldots , \ket{\psi_\lambda(t_m)}\} $ and $\Psi_\lambda(I)$ compare. But first, working with the  idealized subspace $\Psi_\lambda(I)$ we can establish the following inequality chain by Lemma \ref{prelim:check_lem}:
\begin{equation} \label{sample:ineq}
	\Psi_\lambda(I) \leq  \Z(\Ham(\lambda), \ket{\psi_0}) \leq  \A(\Ham)\cdot \ket{\psi_0} ~.
\end{equation}
By the first inequality we see at the outset that our collection of evolution snapshots can at best capture the $\Ham(\lambda)$-cyclic subspace generated by $\ket{\psi_0}$. The ideal scenario is where both of these inequalities are saturated and we now examine under what conditions this is true. First consider the second inequality. Its saturation is purely a property of the Hamiltonian $\Ham(\lambda)$ and the initial state, $\ket{\psi_0}$, as shown by the following lemmas.

\begin{lemma} \label{lem:cyclic_dim}
	The dimension of $\Z(\Ham(\lambda), \ket{\psi_0})$ equals the number of distinct eigenvalues of $\Ham(\lambda)$ with
	eigenspaces that are not orthogonal to $\ket{\psi_0}$.
\end{lemma}

\begin{remark}
	For brevity in the sequel, let $\spec_\lambda$ be the set of distinct eigenvalues for the Hamiltonian $\Ham(\lambda)$ and
	$\eig(\mu)$ be the eigenspace corresponding to $\mu \in  \spec_\lambda$. Then, the lemma takes the following
	symbolic form:
	\[
		\dim \Z(\Ham(\lambda), \ket{\psi_0}) = \#\{\mu \in  \spec_\lambda \st ~ \eig(\mu) \not\perp \ket{\psi_0} \}.
	\]
\end{remark}

\begin{proof}[Proof of Lemma \ref{lem:cyclic_dim}]
	Let $D_\lambda := \diag(\mu_1, \ldots , \mu_d)$ be the diagonal matrix of eigenvalues (counting multiplicity) of $\Ham(\lambda)$.
	Then, by the self-adjointness of $\Ham$, we have that
	\begin{equation} \label{unitary_decomp}
		\Ham(\lambda)^k = U_\lambda D_\lambda^k ~ U_\lambda^\dagger,
	\end{equation}
	for all $k\geq 0$, wherein $U_\lambda$ is a unitary matrix corresponding to the eigenbasis
		$\Eig_\lambda := \{\ket{\xi_1}, \ldots , \ket{\xi_d} \}$
	for $\Ham(\lambda)$. Clearly, the dimension of $\Z(\Ham(\lambda), \ket{\psi_0})$ is given by the size of the maximal linearly
	independent subsets (which we will also call the \emph{rank}) of the set of vectors,
	\begin{align*}
		U_\lambda^\dagger ~ \Z(\Ham(\lambda), \ket{\psi_0}) &= \{U_\lambda^\dagger ~ \Ham(\lambda)^k ~ \ket{\psi_0} \st 0 \leq  k \leq  d-1 \} \\
		&=	\{\left[\mu_1^k ~ \braket{\xi_1}{\psi_0}, \ldots , \mu_d^k ~ \braket{\xi_d}{\psi_0}\right]\trans ~\st~ 0 \leq  k \leq  d-1 \}.
	\end{align*}
	It is immediately seen that
		$\rank(U_\lambda^\dagger ~ \Z(\Ham(\lambda), \ket{\psi_0})) \leq  \#\{\mu \in  \spec_\lambda \st ~ \eig(\mu) \not\perp \ket{\psi_0} \}$ since a completely orthogonal eigenspace will yield a zero vector.
	
	To show the reverse inequality, let $r = \#\{\mu \in  \spec_\lambda \st ~ \eig(\mu) \not\perp \ket{\psi_0} \}$ and consider the matrix $S_\lambda^r$
	whose $(k+1)$st column, for $0\leq k\leq r-1$, is the vector 
		$\left[\mu_{i_1}^k ~ \braket{\xi_1}{\psi_0}, \ldots , \mu_{i_r}^k ~ \braket{\xi_r}{\psi_0}\right] \in  \CC^r$,
	wherein $1\leq j\leq r$, the $\mu_{i_j}$ are distinct and $\ket{\xi_j} \in  \eig(\mu_{i_j}) \cap  \Eig_\lambda$ such that
	$\braket{\xi_j}{\psi_0}\neq 0$. Then, a quick application of the Vandermonde determinant formula shows that
	\[
		\det S_\lambda^r = \prod_{k=1}^d \braket{\xi_k}{\psi_0} ~ \prod_{1\leq s<t\leq d} (\mu_{i_t} - \mu_{i_s}) .
	\]
	Since all of the $\mu_{i_j}$ are distinct and $\braket{\xi_j}{\psi_0} \neq  0$, we see that this determinant is non-zero.
	Moreover, $S_\lambda^r$ is clearly a sub-matrix of the $d \times d$ matrix similarly formed by taking the
	${(k+1)}$st column vector to be $U_\lambda^\dagger ~ \Ham(\lambda)^k ~ \ket{\psi_0}$. Altogether, this means that
	\[
	\rank(U_\lambda^\dagger ~ \Z(\Ham(\lambda), \ket{\psi_0})) \geq  r = \#\{\mu \in  \spec_\lambda \st ~ \eig(\mu) \not\perp \ket{\psi_0} \}
	\]
	and by the
	reverse inequality above, we have that in fact the equality in the statement of the theorem holds.
\end{proof}

This Lemma and the second inclusion in \erf{sample:ineq} together show that
	$\#\{\mu \in  \spec_\lambda \st ~ \eig(\mu) \not\perp \ket{\psi_0} \} \leq  \dim \A(\Ham)\cdot \ket{\psi_0}$
and that this being an equality is both, necessary and sufficient in order that $\Z(\Ham(\lambda), \ket{\psi_0}) = \A(\Ham)\cdot \ket{\psi_0}$. 

This characterization and the method of proof in the Lemma become important ingredients in the following strengthening of \erf{sample:ineq}.
\begin{lemma}
\label{sample:interval}
	For every interval $I \subseteq  \RR$, $\Psi_\lambda(I) = \Z(\Ham(\lambda), \ket{\psi_0})$.
\end{lemma}

\begin{proof}
	In the following we view the state vector as a curve in $d$-dimensional space parameterized by $t$. That is, we consider the dynamics over the time interval $I$, $\ket{\psi_\lambda(t)}:\RR \supseteq  I \rightarrow  \cH_d \cong  \CC^d$, in componentialized form as,
	\[
		\ket{\psi_\lambda(t)} = \left[ \psi_\lambda^1(t), \ldots , \psi_\lambda^d(t) \right] \trans.
	\]
	Then, the dimension of the smallest linear subspace containing $\ket{\psi_\lambda(I)}$ is just the size of maximal
	linearly independent subsets over $I$ of the component functions $\{\psi_\lambda^j : I \rightarrow  \CC \st 1\leq j\leq d\}$.
	Since the dimension is preserved under unitary actions, we begin as in the proof of Lemma \ref{lem:cyclic_dim},
	by taking the unitary diagonalization of $\ket{\psi_\lambda(t)}$ (\ie replacing $D_\lambda^k$ in \erf{unitary_decomp} by
	$\exp(-i t D_\lambda)$), according to the orthonormal eigenbasis $\Eig_\lambda := \{\ket{\xi_1}, \ldots , \ket{\xi_d} \}$
	for $\Ham(\lambda)$, so that we now have the curve,
	\begin{equation} \label{sample:curve}
		U_\lambda^\dagger ~ \ket{\psi_\lambda(t)} = \left[ e^{-i t \mu_1} ~ \braket{\xi_1}{\psi_0}, \ldots , e^{-i t \mu_d} ~ \braket{\xi_d}{\psi_0} \right] \trans.
	\end{equation}
	Thus, analogous to the proof of Lemma \ref{lem:cyclic_dim}, we now consider the question of the size of
	the maximal linearly independent subsets (again, the \emph{rank}) of the set of functions
	\[
		\Gamma_\lambda(I) := \{\braket{\xi_j}{\psi_0} ~ e^{-i t \mu_j} : I \rightarrow  \CC \st 1 \leq  j \leq  d \} .
	\]
	Again, it is easy to see in the one direction that $\rank(\Gamma_\lambda(I)) \leq  \dim(\Z(\Ham(\lambda), \ket{\psi_0}))$.

	Now, letting $r = \dim(\Z(\Ham(\lambda), \ket{\psi_0}))$, take once more the subset $\Gamma_\lambda^r$ of $\Gamma_\lambda$, which contains the
	functions $\gamma_j(t) := \braket{\xi_j}{\psi_0} ~ e^{-i t \mu_{i_j}}$, for each $1\leq j\leq r$, wherein the $\mu_{i_j} $ are
	distinct and $\ket{\xi_j} \in  \eig(\mu_{i_j}) \cap  \Eig_\lambda$ such that $\braket{\xi_j}{\psi_0} \neq 0$. Let
		$\gamma = (\gamma_1, \ldots , \gamma_r) : I \rightarrow  \CC^d$
	and form the Wronskian $W(t) := \det [\gamma(t), \gamma^{(1)}(t), \ldots , \gamma^{(r-1)}(t)]$. Note that
	\[
		\gamma_j^{(k-1)} = (-i \mu_{i_j})^{k-1} ~ \braket{\xi_j}{\psi_0} ~ \exp(-i t \mu_{i_j}) ~ \text{ ($1 \leq  j, k \leq  r$)}
	\]
	so that
	\begin{equation} \label{sample:wronski}
		W(t) = \exp\bigg(-i t \sum_{j=1}^r \mu_{i_j} \bigg) ~ \bigg(\prod_{j=1}^r \braket{\xi_j}{\psi_0} \bigg) ~ 
				\det~\widetilde{V}(-i\mu_{i_1}, \ldots , -i\mu_{i_r}),
	\end{equation}
	where $\widetilde{V}$ is the $r \times r$ Vandermonde matrix. Since
	\[
		\det~\widetilde{V}(-i\mu_{i_1}, \ldots , -i\mu_{i_r}) = \prod_{1\leq j<k\leq r} (-i)(\mu_{i_k} - \mu_{i_j}),
	\]
	and by assumption, $\braket{\xi_j}{\psi_0} \neq  0$ and the $\mu_{i_j}$ are distinct, for all $1\leq j\leq r$, we see that $W(t)\neq 0$
	for all $t \in  I$. Thus, $\Gamma_\lambda^r$ is a linearly independent set of functions, which means also that
		$\rank(\Gamma_\lambda(I)) \geq  \dim(\Z(\Ham(\lambda), \ket{\psi_0}))$,
	and by the reverse inequality above, we see that this degenerates to an equality.
	Hence, the dimension of the minimal linear subspace containing the curve $\ket{\psi_\lambda(I)}$ is the same as that
	of $\Z(\Ham(\lambda), \ket{\psi_0})$, which through the first inequality of \erf{sample:ineq} concludes the proof. 
\end{proof}

\begin{remark}
The two lemmas above aid in identifying a strategy for increasing the effectiveness of time sampling for generating the reduced subspace. Saturation of the first inequality in \erf{sample:ineq} implies that one should try to maximize the dimension of $\Z(\Ham(\lambda), \ket{\psi_0})$. One strategy for doing this is to sample the evolution at multiple parameter values. That is, $\dim( \Z(\Ham(\lambda_1), \ket{\psi_0}) +  \Z(\Ham(\lambda_2), \ket{\psi_0}) )\geq \dim( \Z(\Ham(\lambda_1), \ket{\psi_0}) )$. Hence, sampling the evolution at multiple parameter values will increase the likelihood that the time samples generate the orbit $\A(\Ham)\cdot \ket{\psi_0}$. An interesting and open question is whether it is possible to constructively determine the minimum number of distinct parameter values necessary before the sum of cyclic subspaces equals the orbit $\A(\Ham)\cdot \ket{\psi_0}$. In fact, this is a question that is relevant to any empirical model reduction technique \cite{Antoulas:2009tb, Schilders:2008uc} that attempts to construct reduced order models from snapshots of the state, and is the subject of ongoing work.
\end{remark}

Having established the conditions under which the inequalities in \erf{sample:ineq} are saturated, we return to the question about when the subspace spanned by the time samples $\{\ket{\psi_\lambda(t_1)}, \allowbreak ... ,\allowbreak \ket{\psi_\lambda(t_m)}\}$ with $t_i \in I$ is identical to $\Psi_\lambda(I)$, or equivalently (by Lemma \ref{sample:interval}), to $\Z(\Ham(\lambda),\ket{\psi_0})$. In fact, Lemma \ref{sample:interval} gives us enough data about the dimension of the dynamical curve $\ket{\psi_\lambda(I)}$, locally with respect to time, that we can now show the probabilistic abundance of \emph{good} time samples -- that is, snapshots of the dynamics that span $\Z(\Ham(\lambda), \ket{\psi_0})$.

\begin{theorem}
	Let $r = \dim \Z(\Ham(\lambda),\ket{\psi_0})$. Then, given a vector $\vec{t} \in  \RR^k$ of $r\leq k\leq d$ points in time, collect the
	corresponding state snapshots in
		$\Sample(\vec{t}) := \left[ \ket{\psi_\lambda(t_1)}, \ldots , \ket{\psi_\lambda(t_k)} \right]$.
	This collection $\Sample(\vec{t})$ spans $\Z(\Ham(\lambda), \ket{\psi_0})$ almost everywhere and on a dense and open subset of $\RR^k$.
\end{theorem}

\begin{proof}
	First we prove the openness property. Let $\D_k$ be the collection of functions $\RR^r \rightarrow  \CC$ given by the $r \times r$ minors of $\Sample(\vec{t})$.
	For each $\vec{t} \in  \RR^k$, every $\ell \in  \D_k$ removes some $k-r$ components corresponding to the
	deletion of $k-r$ columns from $\Sample(\vec{t})$. Thus, each $\ell \in  \D_k$ defines a projection
	$p_{\ell} : \RR^k \rightarrow  \RR^r$ that sends each $\vec{t} \in  \RR^k$ to the corresponding $\RR^r$ vector. 
	Since $\vspan(\Sample(\vec{t})) \leq \Z(\Ham(\lambda), \ket{\psi_0})$, $\mathscr{S}(\vec{t})$ spans the
	$r$-dimensional cyclic subspace only if there exists an $r \times r$ minor that is non-zero. Hence the subset
	$\mathscr{T} \subset  \RR^k$ of $\vec{t} \in  \RR^k$ for which $\vspan(\Sample(\vec{t})) = \Z(\Ham(\lambda), \ket{\psi_0})$
	is given by
	\[
		\mathscr{T} = \bigcup_{\ell \in  \D_k} (\ell \circ  p_\ell)^{-1}(\CC \wo 0).
	\]
	Now, $\CC \wo 0$ is an open set and for each $\ell \in  \D_k$, $\ell \circ  p_\ell$ is continuous. 
	Therefore, the collection $\mathscr{T}$ of $k$ time samples that span the cyclic subspace is an open set.

	All that remains is to show is that the complement of $\mathscr{T}$ has Lebesgue measure zero, since in this
	case the density of good snapshots (ones that span the cyclic subspace) is guaranteed. For any $k > r$ tuple of
	times $\vec{t} \in  \RR^k$, the snapshots
		$\Sample(\vec{t})$
	have a span of dimension less than $r$ if and only if the set of vectors corresponding to every $r$-subtuple of
	$\vec{t}$ is linearly dependent. We now show that the set of $r$-tuples of times that give linearly dependent
	state snapshots has Lebesgue measure zero, whence the $k$-tuples of times that give linearly dependent
	$\Sample(\vec{t})$ must also have Lebesgue measure zero.

	Given any tuple of times $\vec{t} \in  \RR^r$, the linear independence of the snapshots
		$\{\ket{\psi_\lambda(t_1)}, \allowbreak \ldots , \allowbreak\ket{\psi_\lambda(t_r)} \}$
	is equivalent to the process of inductively checking that for each $1<j\leq r$,
		$\ket{\psi_\lambda(t_j)} \not\in  \vspan \{\ket{\psi_\lambda(t_1)}, \ldots , \ket{\psi_\lambda(t_{j-1})} \}$.
	So suppose that we have $\vec{t} \in  \RR^j$ for $r > j \geq  1$ such that the corresponding state snapshots are
	linearly independent, and consider adding a new state snapshot $\ket{\psi_\lambda(t)}$ to this list. Now, since each
	$(j+1) \times (j+1)$ minor of the matrix
	\[
		\left[ \ket{\psi_\lambda(t_1)} ~ \cdots  ~ \ket{\psi_\lambda(t_j)} ~ \ket{\psi_\lambda(t)} \right]
	\]
	is a polynomial in $j+1$ analytic functions $\psi^{i}_\lambda(t)$ (elements of $\ket{\psi_\lambda(t)}$), it defines an analytic
	curve in $t$. Thus it has only countably many roots or is identically zero \cite[\textsection~3.1]{krantz2002primer}. By
	Lemma \ref{sample:interval}, the latter case cannot be achieved. Therefore there are only a countably number of
	time points $t$ such that the $(j+1) \times (j+1)$ minors are zero, and hence for almost all $t$,
		$\ket{\psi_\lambda(t)} \not\in  \vspan \{\ket{\psi_\lambda(t_1)}, \ldots , \ket{\psi_\lambda(t_j)} \}$.
	Whence by induction, we further have that for almost all $\vec{t} \in  \RR^r$, the snapshots
		$\{\ket{\psi_\lambda(t_1)}, \ldots , \ket{\psi_\lambda(t_r)}\}$
	will be linearly independent.
\end{proof}

The probabilistic abundance of time samples that span the cyclic subspace provides strong support for empirical model reduction techniques that construct the reduced order model from sampling the solution of the full order model. However, in practice the time evolution is often not randomly sampled, but rather state snapshots are taken at regularly spaced time intervals, $\Delta t$. The following Proposition specifies conditions under which this strategy provides good state samples.

\begin{proposition} \label{sample:arith}
	Fix a sampling time $\Delta t \in  \RR$ and let $r = \dim(\Z(\Ham(\lambda), \ket{\psi_0}))$. Then, for any $m \in  \ZZ$, the state samples
	corresponding to the sequence of time-steps, 
	\[
		\vec{t} = (m \Delta t, (m+1) \Delta t, \ldots , (m+r-1) \Delta t),
	\]
	span $\Z(\Ham(\lambda), \ket{\psi_0})$ if and only if for each $1 \leq  j < k \leq  r$, whenever $\mu_k \neq  \mu_j$,
	\[
		\Delta t \not\equiv  0 ~ \left( \emph{mod}~ \frac{2\pi}{\mu_k - \mu_j} \right).
	\]
\end{proposition}

\begin{proof}
	Consider the matrix formed by applying each element of $\vec{t}$ to \erf{sample:curve} and let
	$V_m$ be the $r \times r$ sub-matrix given by removing any zero or duplicate rows (identically in $\Delta t$). This 
	matrix then takes the form,
	\[
		V_m(z_1, \ldots , z_r) :=
		\begin{pmatrix}
		w_1 z_1^m			& \cdots & w_1 z_1^{m+r-1} \\
		\vdots			& \ddots & \vdots \\
		w_r z_r^m		& \cdots & w_r z_r^{m+r-1}
		\end{pmatrix} ,
	\]
	with $z_k = e^{-i\Delta t \mu_k}$ and $w_k = \braket{\xi_k }{ \psi_0}$. Analogous to the Wronskian calculation in
	\erf{sample:wronski}, we get
	\[
		\det V_m(z_1, \ldots , z_r) = \prod_{j=1}^r w_j z_j^m ~ \prod_{1\leq j<k\leq r} (z_k - z_j).
	\]
	Since the removal of zero columns ensures $w_j \neq  0$ and the removal of duplicate rows guarantees
	that the $\mu_j$ are distinct, $\det(V_m)$ is clearly non-zero if and only if
	$\Delta t \equiv  0 \pmod{\sfrac{2\pi}{\mu_k - \mu_j}}$ for all $1 \leq  j < k \leq  r$. 
\end{proof}
\begin{remark}
The physical interpretation of this condition is that the sampling period should not overlap with any \emph{intrinsic periods} of the system, where the latter are defined as the inverse of the Bohr frequencies of the Hamiltonian.
\end{remark}

\subsection{Computational feasibility}
Construction of the reduced subspace of $\Ham(\lambda)_{\ket{\psi_0}}$ is an ``offline" computation that only needs to be done once. It then allows for more efficient simulation of the model for all values of the parameters $\lambda$. However, it is instructive to examine the computational complexity of the procedures we have outlined above for constructing the model reduction matrix $\Phi$.

First consider the construction via the Burnside basis as detailed in \textsection  \ref{sec:cert_general} and Appendix \ref{app:burnside_alg}. Implementations of Algorithm \ref{alg_burnside} for constructing the Burnside basis will involve multiplication of matrices and rank computations. Subsequently, a basis for the reduced subspace is constructed by matrix-vector multiplication. All of these are polynomial complexity operations in the size of the matrices and vectors, $d$. However, in the quantum context, $d$ grows exponentially in the number of degrees of freedom and therefore this algorithm is in general going to be computationally infeasible for large quantum systems. Despite this, in certain cases the structure of the underlying matrices can be exploited to formulate more efficiently implementable versions of the Burnside algorithm, as we shall see in \textsection  \ref{sec:pauli}.

Next, consider the method for constructing the reduced subspace by sampling the time evolution of the system. The results in \textsection  \ref{sec:sample} specify the conditions required for such time samples to span the reduced subspace. Obtaining the time samples requires solution of the full order model, and then a determination of a basis for the space spanned by the samples. Both tasks have complexity that scales exponentially in the number of degrees of freedom due to the exponential scaling of the vector dimension $d$. However, this procedure can be useful if a small number of snapshots can be generated by large-scale simulations at a single, or few, parameter values $\lambda$. Then the model reduction can be performed to more efficiently explore the parametric dependence of the model.

\subsection{Robustness of dimensional reductions}
The structural aspects of $\A(\Ham)$ as discussed in \textsection ~ \ref{sec:struct_A} allow us now to study how \emph{much} model reduction is available -- that is, how large is the minimal number of dimensions required to contain the dynamics -- for either perturbations of the initial state $\ket{\psi_0}$ or \emph{generically}, that is, when the initial state is chosen at random. Note firstly that the decomposition \erf{red:fulldecomp} has also the associated decomposition of the dimension of the system space:
\begin{equation} \label{eq:dimdecomp}
	d = j_1 q_1 + \cdots + j_m q_m .
\end{equation}
Using the unitary $U$ that gives the decomposition \erf{red:fulldecomp}, we begin by partitioning the new basis representation $\ket{\nu} := U \ket{\psi_0}$ of the initial state into $j_1 + \cdots + j_m$ vector components: define the sequence $J_0 = 0$ and $J_k = j_1 + \cdots + j_k$ (for $1 \leq k \leq m$) and take the partition
\begin{equation} \label{eq:ketpart}
	\ket{\nu} = [\ket{\nu_1}, \ldots, \ket{\nu_{J_1}}, \ket{\nu_{J_1 + 1}}, \ldots, \ket{\nu_{J_2}}, \ldots, \ket{\nu_{J_m}}]^{\trans}
\end{equation}
such that for each $\alpha \in \{ J_{k-1} + 1, \ldots, J_k \}$, $\Pvec \in \CC^{q_k}$. For purposes of exposition, we denote by $q(\alpha)$ the dimension $q_k$ of $\Pvec$ and by $j(\alpha)$ the corresponding factor $j_k$ in \erf{eq:dimdecomp}; see Figure \ref{fig:mat_decomp} for an example that illustrates this notation and partitioning of $\ket{\nu}$. Now, since $\dim(\Mat_{q(\alpha)}(\CC)\cdot \Pvec) = q(\alpha)$ or $0$ according to whether or not $\Pvec \neq 0$, respectively, we see immediately that a sufficiency condition for $\A(\Ham)\cdot \ket{\psi_0} = U^* (\A(\Ham)\cdot \ket{\nu}) < \cH_d$ is that some $\Pvec = 0$. This is not a necessary condition as we will see below.

\begin{figure}[h!]
\centering
{\includegraphics[scale=0.35]{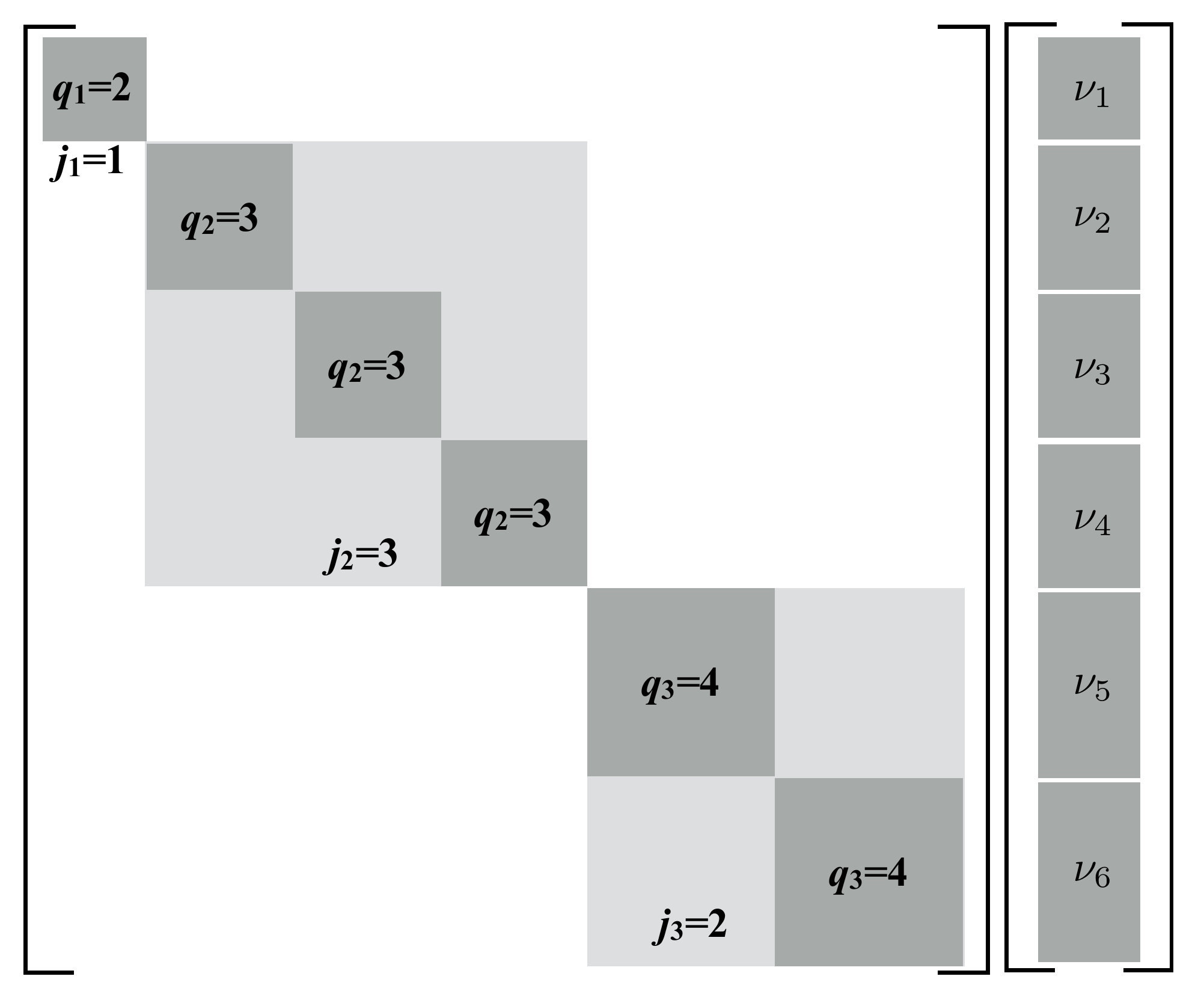}}
\caption{\small \label{fig:mat_decomp} An example of the block decomposition of elements of $\A(\Ham)$, as prescribed by \erf{red:fulldecomp} and a $d$-dimensional vector written in the basis that achieves this block-diagonal form. The only non-zero elements are in the dark gray sub matrices. In this case, $d=19, m=3$, and the $J_k$ as defined in the main text are: $J_1=1, J_2=4,J_3=6$.}
\end{figure}

\FloatBarrier

We can actually describe exactly the dimension of the reduced subspace in terms of the new representation of the initial state. To this end, the following lemma describes first the case for an individual component $\Mat_q(\CC) \otimes \Unit_j$ in \erf{red:fulldecomp}.

\begin{lemma}
	Let $\ket{\nu} \in \CC^{jq}$ for some $q, j \geq 1$ and partition it into $j$ sub-vectors, each of length $q$. Further, denote by $M$ the $q \times j$ matricial form of this partition of $\ket{\nu}$. Then,
	\[
		\dim ~ (\Mat_q(\CC) \otimes \Unit_j)\cdot \ket{\nu} = q \rank(M) .
	\]
\end{lemma}

\begin{proof}
	Clearly, $(\Mat_q(\CC) \otimes \Unit_j)\cdot \ket{\nu} \cong \Mat_q(\CC)\cdot M$, by vectorization of elements in the space on the right-hand side. Now, let $\beta := \{e_{k,l} \}_{1 \leq k,l \leq q}$ be the canonical basis for $\Mat_q(\CC)$ such that the $(k,l)$ entry of $e_{k,l}$ is $1$ (and the matrix is zero everywhere else) and note that $\Mat_q(\CC) \cdot M = \vspan\{e M \st e \in \beta \}$. We see further that the $l$th row of $e_{k,l} M$ is the $k$th row of $M$ and the matrix is zero everywhere else. Thus, for each fixed $l$, there are exactly $\rank(M)$ linearly independent matrices $e_{k,l} M$ for $1 \leq k \leq q$, since this is just the number of linearly independent rows of $M$. By the $q$ many choices for $l$ (due to the number of rows in $M$), we have the desired dimensionality computation in the statement of the lemma.
\end{proof}

This lemma along with the decomposition \erf{red:fulldecomp} give immediately the following form for the dimension of orbit of the initial state under $\A(\Ham)$:

\begin{theorem} \label{thm:diminv}
	Let $\ket{\nu}$ be as in \eqref{eq:ketpart} and define the matrix $M_k = \left[\ket{\nu_{J_{k-1}+1}} ~ \cdots ~ \ket{\nu_{J_k}} \right] \in \Mat_{q_k \times j_k}(\CC)$.
	Then,
	\begin{equation} \label{eq:diminv}
		\dim(\A(\Ham)\cdot \ket{\psi_0}) = \sum_{k=1}^m q_k \rank(M_k) .
	\end{equation}
\end{theorem}

Now we are in a position to assess how perturbations to the initial state affect the dimension of the reduced order model. Consider a Hamiltonian model $\Ham_{\ket{\psi_0}}$ that admits a reduced subspace of dimension $r<d$ -- \ie the model reduction matrix $\Phi \in \Mat_{d \times r}(\CC)$. Writing $\ket{\nu} = U\ket{\psi_0}$, we have by the assumption $r < d$ and \erf{eq:diminv} that for some $1\leq k\leq m$, $\rank(M_k) < \min(j_k, q_k)$. The theorem further shows us how the dimension of the reduced order model will change under a perturbation of the initial state to $\ket{\psi_0'} \in \cH_d$. Indeed, write $\ket{\nu'} = U\ket{\psi_0'}$ and let $M_k = \left[\ket{\nu_{J_{k-1}+1}} ~ \cdots ~ \ket{\nu_{J_k}} \right]$ and $M_k' = \left[\ket{\nu_{J_{k-1}+1}'} ~ \cdots ~ \ket{\nu_{J_k}'} \right]$. Then, $\rank(M_k')$ can be less or greater than $\rank(M_k)$ for $1 \leq k \leq m$ and Theorem \ref{thm:diminv} states that these ranks of the perturbed initial state dictate the change in the dimension of the reduced order model. The most obvious form of \emph{instability} in the reduced order model dimension under a perturbation of the initial state occurs when $M_k=0$ for some $k$, so that any perturbation in only that component of the vector (\ie so that $\rank(M_k')>0$ and $\rank(M_l)=\rank(M_l') ~~~\forall l\neq k$) increases the dimension of the reduction by $q_k \rank(M_k')$. Interestingly, Theorem \ref{thm:diminv} points out that there exist certain perturbations that have minimal change to the model reduction dimension. These are ones that possess symmetry that leads to significant linear dependencies between $\ket{\nu_{J_{k-1}+1}}, \ldots, \ket{\nu_{J_k}}$ for some $k$, and therefore $\rank(M_k')$ remains small. However, if the perturbation to the initial state is a randomly drawn pure state (\ie random according to the pushforward measure onto $\cH_d$ of the Haar measure on the group $U(d)$ of unitaries), then one expects that with overwhelming probability the dimension of the reduced subspace will increase, making the reduced order model dimensionally not robust to such perturbations in the initial vector. Therefore such robustness to perturbations in the initial state depend strongly on the form of the perturbing state, but regardless of the form, the representation of the perturbation in the basis that block diagonalizes the algebra and Theorem \ref{thm:diminv} allow one to assess the robustness.


\section{Pauli Hamiltonians}
\label{sec:pauli}
Until now the Hamiltonian model and state-space that it acts on have been completely general (apart from the assumption of finite dimension). In the following we will specialize to interacting spin-$\frac{1}{2}$ (or qubit) models. In this setting we can exploit properties of the generalized Pauli algebra to strengthen our results. This subset of models is also of considerable interest to the condensed matter and quantum information communities. Explicitly, the setting in this section is: $\cH_d \cong (\CC^2)^{\otimes n}$ with $d=2^n$, where $n$ is the number of spins in the system. 

\subsection{Preliminaries}
Let $\Ham(\lambda) = \sum_{k=1}^M \lambda_k H_k$ be an $n$-spin Hamiltonian model, meaning that each $H_k$ is a linear combination of elements of $\Pb_n^*$. Writing each $H_k$ out explicitly in the basis of Pauli operators, $\Pb_n^*$, we get
\begin{equation} \label{eq:pauli_ham}
	\Ham(\lambda) = \sum_{k=1}^M \lambda_k \left( \sum_{j=1}^K \alpha_{k,j} ~ \bp_j^k \right) ,
\end{equation}
where each $\bp_j^k \in \Pb_n^*$, and the $\alpha_{i,j}$ are scalar coefficients appearing in the linear expansion of the component Hamiltonian $H_k$. Note that the number of parameters, $M$, defining the model, does not have to equal the number of Pauli operators in this expansion -- one parameter could \emph{control} multiple Paulis. We collect the set of Pauli operators appearing in the description of the model $\Ham$, in the following:

\begin{definition}
$\Sb_\lambda(\Ham) := \{\sigma \in  \Pb_n^* \st \tr(\sigma^\dagger \Ham(\lambda)) \neq  0\}$ is the set containing the Pauli matrices from $\Pb_n^*$ whose coefficients, $\alpha_{k,j}$, in the decomposition \erf{eq:pauli_ham} of $\Ham(\lambda)$ are non-zero.
\end{definition}

When each parameter of the Hamiltonian model controls exactly one Pauli operator, we are in a special case, where many calculations considered can be greatly simplified. We single out this case as:

\begin{definition}
\emph{Pure Pauli Hamiltonians} are spin Hamiltonians where $|\Sb_\lambda(\Ham)| = M$. That is, each Pauli operator is multiplied by an independent free parameter in the model:
	\begin{equation} \label{eq:pure_pauli}
		\Ham^*(\lambda) = \sum_{k=1}^M \lambda_k \bp_k ,
	\end{equation}
	where again, $\bp_k \in \Pb_n^*$.
\end{definition}

In fact, in some cases it may be beneficial to over-parametrize a Hamiltonian of the form \erf{eq:pauli_ham} by a pure Pauli Hamiltonian by considering each coefficient as being an independent parameter for the model. When this is done, we will refer to it as \emph{over-parameterizing} a model. For example, consider an $n$-spin 1D transverse field Ising model with periodic boundary conditions, which has the Hamiltonian
\[
	\Ham(\lambda= (B,J)) = -B \sum_{j=1}^n \sigma_x^{(j)} - J \left(\sum_{j=1}^{n-1} \sigma_z^{(j)} \otimes \sigma_z^{(j+1)} + \sigma_z^{(n)} \otimes \sigma_z^{(1)}\right),
\]
where we have suppressed identities in writing $n$-spin Pauli operators, \eg $\sigma_x^{(j)} := I_2 \otimes ... \otimes \sigma_x \otimes ... \otimes I_2$, where the $\sigma_x$ is in the $j^{th}$ position. The over-parameterized Hamiltonian corresponding to this model is:
\[
	\Ham^*(\Lambda) = \sum_{j=1}^n \Lambda_j \sigma_x^{(j)} + \sum_{j=1}^{n-1} \bar \Lambda_j\sigma_z^{(j)} \otimes \sigma_z^{(j+1)} + \bar \Lambda_n\sigma_z^{(n)} \otimes \sigma_z^{(1)}
\]
That is, each Pauli term in the original Hamiltonian is multiplied by an independent parameter. The reduced subspaces for $\Ham^*(\Lambda)_{\ket{\psi_0}}$ contain the minimal reduced subspace for $\Ham(\lambda)_{\ket{\psi_0}}$, and consequently $\dim \A(\Ham^*)_{\ket{\psi_0}} \geq \dim \A(\Ham)_{\ket{\psi_0}}$. However, we shall see below that computing the \emph{minimal} such over-approximation involves a much more efficient procedure than computing the minimal reduced subspace for $\Ham$. 

The major computational simplification in the case of pure Pauli Hamiltonians, which we shall exploit in the following, is provided by an isomorphism between the elements of $\Pb_n^*$ and binary vectors of length $2n$:
\begin{theorem} \label{pauli:groups}
	Let $B = \{\pm i I_d, \pm  I_d \}$, and $G_n := \langle \Pb_n^*\rangle $ be the $n$-spin Pauli group. Then, there is a bijection $G_n / B \leftrightarrow \Pb_n$ that naturally induces a group structure on $\Pb_n$ (so that $G_n / B \cong  \Pb_n$). Furthermore, under this group structure, there exists an isomorphism $\varphi:\Pb_n \xrightarrow{\sim} (\ZZ_2)^{2n}$.
\end{theorem}
\begin{proof}
	Since $G_n = \{\pm  i \bp, \pm  \bp \st \bp \in  \Pb_n \}$ and given $\bp \in  \Pb_n$, $\bp B = \{\pm  i \bp, \pm  \bp \}$, we have that the mapping $\Pb_n \ni \bp \mapsto  \bp B \in  G_n / B$ is clearly a bijection. By defining multiplication in $\Pb_n$ to be matrix multiplication \emph{modulo} $B$, we furthermore have a group structure on $\Pb_n$.

	Under this group structure, $\Pb_n$ is abelian and each element is of order two, making this is an \emph{elementary abelian} group, which by the classification theorem for finitely generated abelian groups \cite[Theorem 4.5.1]{herstein1975topics} and the fact that $|\Pb_n| = 2^{2n}$, gives us the
	existence of an isomorphism $\varphi:\Pb_n \xrightarrow{\sim} (\ZZ_2)^{2n}$.
\end{proof}

\begin{remark}
	The identification $\Pb_n \cong  (\ZZ_2)^{2n}$ is commonly made in quantum error correction literature
	(see \cite[\textsection~10.2]{lidar2013quantum} and \cite{calderbank1998quantum}) in a more constructive
	manner, by extending $\varphi$ tensorially from the $n=1$ case, wherein $\sigma_x \mapsto  (0,1), \sigma_y \mapsto  (1,1)$ and
	$\sigma_z \mapsto  (1,0)$.
\end{remark}

The isomorphism $\varphi$ induces what is referred to as a \emph{binary vector representation} of the Pauli group generators \cite[\textsection~2]{calderbank1998quantum}. This enables us to represent each element of $\Pb_n$ as a binary vector of length $2n$ (linear in the number of spins), and the multiplication operation among matrices is mapped to the XOR operation among binary vectors.

\subsection{Certificate for invariant subspaces of Pauli Hamiltonians}
The algebraic certificate presented in \textsection  \ref{sec:cert_general} for the existence of invariant subspaces for a Hamiltonian model $\Ham$ requires computing a basis for the subalgebra generated by $\Coeff(\Ham)$. The bottlenecks in this computation (detailed in Appendix \ref{app:burnside_alg}) are (1) the multiplication of elements in $\Coeff(\Ham)$, and (2) $\vspan$ checks to assess linear dependence of elements. For general Pauli Hamiltonians of the form \erf{eq:pauli_ham} neither of these bottlenecks are completely removed (although, since the Pauli matrices have significant structure the multiplications should be performed symbolically as opposed to explicitly forming the $d\times d$ matrices). In contrast, considerable simplification is possible for \emph{pure Pauli} Hamiltonians of the form \erf{eq:pure_pauli} and both bottlenecks can be overcome. In this case each element of the natural $\Coeff(\Ham)$ is a member of $\Pb_n$ and has a representation as a vector in $(\ZZ_2)^{2n}$. Further, multiplication of two elements in $\Pb_n$ results in another element in $\Pb_n$ (along with a multiplicative constant that we do not need to track), and thus all multiplications in Algorithm \ref{alg_burnside} in Appendix \ref{app:burnside_alg} can be implemented as binary addition of vectors of length $2n$ (as opposed to multiplication of matrices of size $2^n \times 2^n$). As for bottleneck (2), the fact that all elements of $\Pb_n$ are linearly independent, implies that any linear dependence or $\vspan$ based checks in the algorithm become unnecessary. Instead, we need only check for duplicates in the collection of multiplicatively generated elements (since computational storage objects such as arrays are forms of \emph{multisets} and not the usual mathematical sets), the removal of which gives our Burnside basis (the test for membership can be efficiently implemented by the AND operation on binary vectors). A straightforward modification of the algorithm to generate the Burnside basis that exploits these simplifications to simplify computations is specified as Algorithm \ref{alg:pauli_burnside} in Appendix \ref{app:burnside_alg}. It is also possible to construct a completely different procedure that utilizes all the structure in $\Pb_n$ and its binary vector representation, and this is presented as Algorithm \ref{alg:pauli_burnside_basis_2} in Appendix \ref{app:burnside_alg}. It greatly improves the efficiency of the Burnside basis calculation by utilizing the key structural properties of $\Pb_n$: that it is an elementary abelian $2$-group and all of its operators are linearly independent as vectors of $\Mat_d(\CC)$. This is done by quickly enumerating all possible binary vector additions through a Gray codes table, a method employed in the Method of Four Russians algorithm for quickly generating subspaces of $(\ZZ_2)^k$ (see \cite[\textsection~9.2]{bard2009algebraic} and \cite[\textsection~2]{bard2006accelerating}).
All required operations in Algorithms \ref{alg:pauli_burnside} and \ref{alg:pauli_burnside_basis_2} are clearly polynomial complexity in $n$, the number of spins. However, it must be noted that the size of the Burnside basis, $|\mathscr{B}(\Ham)|$, can still be exponential in $n$.

As alluded to above, this simplification for pure Pauli Hamiltonians motivates the over-parametrization of spin Hamiltonians; even if the true Hamiltonian takes the form in \erf{eq:pauli_ham} one may want to over-parametrize it to obtain the form \erf{eq:pure_pauli} because computing the algebraic certificate for the existence of invariant subspaces for $\Ham^*$ is more efficient. However, one must be mindful that invariant subspaces for the model $\Ham$ may exist even when none exist for $\Ham^*$.

In addition to simplifying the calculation of the Burnside basis (whose size provides a necessary and sufficient condition for model reduction by Proposition \ref{red:burnprop}), the $\Pb_n \cong  (\ZZ_2)^{2n}$ isomorphism allows us to derive a strong \emph{sufficiency} criteria for Pauli Hamiltonians:
\begin{theorem} \label{pauli:thm1}
	Any $n$-spin Hamiltonian $\Ham(\lambda)$ with $|\Sb_\lambda(\Ham)| < 2n$ has a non-trivial proper invariant subspace.
\end{theorem}

\begin{proof}
We begin by noting $\A(\Ham) \subseteq  \A(\Sb_\lambda(\Ham))$. Then note that $\A(\Sb_\lambda(\Ham))$ is just the linear span of $\langle \Sb_\lambda(\Ham)\rangle _B$, where $\langle X\rangle _B$ denotes the group generated by $X$, modulo $B$. Therefore we have by Theorem \ref{red:burnside} that if $\vspan(\langle \Sb_\lambda(\Ham)\rangle _B)$ is a proper subspace of $L(\cH_d)$ (recall that $\cH_d = (\CC^2)^{\otimes n}$), then $\A(\Sb_\lambda(\Ham))$ is reducible. These facts together with Proposition \ref{red:burnprop} imply that if $\vspan(\langle \Sb_\lambda(\Ham)\rangle _B) < L(\cH_d)$, then the model $\Ham$ has a non-trivial proper invariant subspace. 

We know that $L(\cH_d)$ is spanned by the linearly independent elements of $\Pb_n$, and so $\vspan(\langle \Sb_\lambda(\Ham)\rangle _B) < L(\cH_d)$ is equivalent to $\langle \Sb_\lambda(\Ham)\rangle _B < \Pb_n$, \ie the group generated by $\Sb_\lambda(\Ham)$ is a proper subgroup of the generators of the Pauli group. To show that this is true when $|\Sb_\lambda(\Ham)| < 2n$, we turn to the binary vector representation of $\Pb_n$: $(\ZZ_2)^{2n}$ is a $2n$-dimensional linear space over $\ZZ_2$ with scalar multiplication given in the natural way: for any $\bm{\nu} \in  (\ZZ_2)^{2n}$, $0\bm{\nu} = 0$ and $1\bm{\nu} = \bm{\nu}$. As a result, constructing the linear span of elements in $(\ZZ_2)^{2n}$ is equivalent to constructing the group generated by the corresponding elements in $\Pb_n$. Moreover, since the $2n$-dimensional linear space $(\ZZ_2)^{2n}$ is spanned by $2n$ binary vectors, we can conclude that the minimal generating set of the group $\Pb_n$ is of size $2n$; \ie $\rank(\Pb_n) = 2n$. This is sufficient to prove the result, since by assumption $|\varphi(\Sb_\lambda(\Ham))| = |\Sb_\lambda(\Ham)| < 2n$, and therefore $\text{span}(\varphi(\Sb_\lambda(\Ham))) = \varphi(\langle \Sb_\lambda(\Ham)\rangle _B) < (\ZZ_2)^{2n}$, whence $\langle \Sb_\lambda(\Ham)\rangle _B < \Pb_n$.
\end{proof}

\begin{remark}
If $|\Sb_\lambda(\Ham)| = 2n-m$ for $0 < m \leq 2n$, then a straightforward bound on the size of the Burnside basis is given by $|\mathscr{B}(\Ham)| \leq \frac{4^n}{2^m}$.
\end{remark}

A generalization of the above result to the complementary case where $|\Sb_\lambda(\Ham)| \geq  2n$ is possible by again utilizing the isomorphism $\varphi:\Pb_n \xrightarrow{\sim} (\ZZ_2)^{2n}$ that allows us to associate a collection of Pauli operators with a binary matrix. But in this case a counting argument is insufficient and one needs to perform a computation:
\begin{theorem} \label{pauli:thm2}
	An $n$-spin Hamiltonian model $\Ham$ has a non-trivial proper invariant subspace whenever $\rank \varphi(\Sb_\lambda(\Ham)) < 2n$. 
\end{theorem}

\begin{proof}
	Clearly, the binary matrix corresponding to the collection of vectors $\varphi(\Sb_\lambda(\Ham)) \subseteq  (\ZZ_2)^{2n}$ has rank less than $2n$ if and only if $\im(\varphi|_{\langle \Sb_\lambda(\Ham)\rangle }) < (\ZZ_2)^{2n}$. By the proof of Theorem \ref{pauli:thm1}, this implies that $\Ham$ has a non-trivial proper invariant subspace. 
\end{proof}

\vspace{2mm}
\begin{remarks}\begin{multirem}
\rem The rank computation required in this theorem is performed over the field $\ZZ_2$ and thus has some efficient implementations based on the Method of Four Russians \cite{M4RI}.
\rem For pure Pauli Hamiltonians, this condition becomes both necessary and sufficient for the existence of a proper non-trivial invariant subspace for the model. If a basis for the reduced subspace must be constructed, then the Burnside basis must be explicitly calculated, which is nothing more than the group generated by $\Sb_\lambda(\Ham)$ (see Algorithms \ref{alg:pauli_burnside} and \ref{alg:pauli_burnside_basis_2} in Appendix \ref{app:burnside_alg}).
\end{multirem}\end{remarks}

Theorems \ref{pauli:thm1} and \ref{pauli:thm2} are useful sufficient conditions for Pauli Hamiltonians: by simply counting the number of terms in the Pauli decomposition of a spin Hamiltonian (or by computing the rank of a binary matrix in the case of Theorem \ref{pauli:thm2}) they allow one to check for the existence of invariant subspaces for the model. However, it should be noted that these sufficient conditions can be very conservative if $M < |\Sb_\lambda(\Ham)|$. In this case, the number of free parameters is smaller than the number of Pauli ``directions" in the Hamiltonian and an invariant subspace can exist even if the sufficient conditions in Theorems \ref{pauli:thm1} and \ref{pauli:thm2} are not met.

\subsection{Constructing reduced subspaces for Pauli Hamiltonians}
\label{sec:construct_pauli}
Once the Burnside basis has been computed, the remaining step is to apply its elements to $\ket{\psi_0}$ and collect the maximal linearly independent members of the resulting vectors into the model reduction matrix $\Phi$. For general Pauli Hamiltonians this is just as computationally demanding as for general Hamiltonians. Even for pure Pauli Hamiltonians, where each element of the Burnside basis $B_j \in \Pb_n$, this step is normally computationally demanding since $\ket{\beta_j} = B_j \ket{\psi_0}$ cannot be computed using the binary vector representation of $B_j$ in general.

Nonetheless, in the following we sketch an efficient approach to computing the reduced order model in the special case where expectation values of Pauli operators under the initial state are easily computable; explicitly, $\expect{\bp}  := \bra{\psi_0} \bp \ket{\psi_0}$ is assumed to be known or efficiently computable for all $\bp \in \Pb_n$. This assumption is valid for many physically relevant scenarios where the initial state has particularly simple structure, \eg $\ket{\psi_0}$ is a separable state or a matrix product state \cite{Verstraete:2008ko}. Given this assumption and the fact that $B_j \in \Pb_n$ for pure Pauli Hamiltonians, a maximal linearly independent subset of $\{\ket{\beta_j}\}_{j=1}^{|\mathscr{B}(\Ham)|}$ can be efficiently found by the following procedure:
\begin{enumerate}
\item Initialize the linearly independent subset as $ \mathscr{L} = \{\ket{\beta_1} \}$.
\item Take a vector $\ket{\beta_j}, ~~ 1<j\leq |\mathscr{B}(\Ham)|$, and compute the Gramian matrix for $\ket{\beta_j} \cup \mathscr{L}$. This is efficient by assumption since the entries of the Gramian matrix are all of the form $\bra{\psi_0}B_k\dg B_l \ket{\psi_0} = \expect{\bp}$ for some $\bp \in \Pb_n$.
\item If the resulting Gramian matrix has nonzero determinant, include $\ket{\beta_j}$ in $\mathscr{L}$.
\item Repeat steps 2-3 for all $1<j\leq |\mathscr{B}(\Ham)|$.
\end{enumerate}
At the conclusion of this procedure the columns of the model reduction matrix $\Phi$ are the elements of $\mathscr{L}$ (after suitable ortho-normalization).

As with the computation of the Burnside basis, the above procedure is not guaranteed to be polynomial complexity in the number of spins, $n$, since $|\mathscr{B}(\Ham)|$ could be exponential in $n$. However, for models where significant model reduction is possible the dimension of the Gramian matrix whose determinant needs to be computed will not grow quickly since although $|\mathscr{B}(\Ham)|$ can be large, the number of linearly independent vectors in the set $\{\ket{\beta_j}\}_{j=1}^{|\mathscr{B}(\Ham)|}$ will be small.

The assumption of easily computable Pauli expectation values under the initial state also enables efficient formation of the reduced order Hamiltonian in \erf{eq:red_order_ham}. Forming this Hamiltonian by brute-force requires the projection of a $d\times d$ matrix, but writing out this projection in the Pauli case yields,
\[
\calhat{\Ham}(\lambda) = \Phi \dg \Ham(\lambda) \Phi = \sum_{i=1}^M \lambda_i \left[\begin{array}{c} \bra{\psi_0}B_{j_1} \\ \bra{\psi_0} B_{j_2} \\ \vdots \\ \bra{\psi_0} B_{j_r} \end{array}\right] \bp_i \left[ B_{j_1} \ket{\psi_0}, B_{j_2} \ket{\psi_0}, ..., B_{j_r} \ket{\psi_0}\right ],
\]
where the columns of $\Phi$ are formed from the $r$ maximal linearly independent set of vectors out of $\{ B_j \ket{\psi_0}\}_{j=1}^{|\mathscr{B}(\Ham)|}$. Recalling that $B_j \in \Pb_n$, it is clear that each element of the matrix $\calhat{\Ham}(\lambda)$ takes the form $\bra{\psi_0}B_k \bp_i B_j \ket{\psi_0} = \expect{\bp}$ for some $\bp \in \Pb_n$ that can be determined efficiently since the product $B_k \bp_i B_j$ can be calculated using the binary vector representation. 

\section{Examples}
\label{sec:eg}
In this section we apply the methods developed thus far to three paradigmatic spin models. In all the examples, we notate spin states using the $\sigma_z$ eigenbasis $\{\ket{0}, \ket{1}\}$, with $\sigma_z \ket{b} = (-1)^{b}\ket{b}$.

\begin{example}[Collective rotation Hamiltonian]
A spin model with significant symmetry commonly used to describe nuclear magnetic resonance systems is the collective rotation Hamiltonian
\begin{equation}
\label{eq:eg_coll_spin}
		\Ham_{\textrm{$n$-spin}}(\lambda) = \lambda_z \sum_{j=1}^n \sigma_z^{(j)} + \lambda_x \sum_{j=1}^n \sigma_x^{(j)} + \lambda_y \sum_{j=1}^n \sigma_y^{(j)},
\end{equation}
where $j$ denotes the spin on which the Pauli matrix acts non-trivially, and the three model parameters are $\lambda = (\lambda_x, \lambda_y, \lambda_z)$. This dynamical model possess complete permutation symmetry and therefore we would expect a simple description of the dynamics if the initial state is also permutation symmetric, regardless of the values of $\lambda_{x,y,z}$. Table \ref{tab:coll_spin} lists the size of the Burnside basis, $|\mathscr{B}(\Ham)|$, and $r = \dim \vspan ~\{B_j\ket{\psi_0}\}_{j=1}^{|\mathscr{B}(\Ham)|}$, as a function of number of spins for both a permutationally symmetric and asymmetric initial state.

For any number of spins, except for $n=1$, a non-trivial invariant subspace exists since $|\mathscr{B}(\Ham)| < 4^n$. Furthermore, the reduced subspace dimension appears to scale linearly for the permutionally invariant initial state. Even in the case of an initial state that is not completely permutation symmetric, $\ket{\psi_0} = \ket{1}\ket{0}^{\otimes n-1}$, the scaling of $r$ appears to be linear. 

\begin{table*}\centering
\ra{1.3}
\begin{tabular}{@{}lrrrrrr@{}}
\toprule
$n$ & 1 & 2 & 3 & 4 & 5 & 6 \\
\midrule
$|\mathscr{B}(\Ham)|$ & 4 & 10 & 20 & 35 & 56 & 84 \\
$r$ ($\ket{\psi_0} = \ket{0}^{\otimes n}$) & 2 & 3 & 4 & 5 & 6 & 7 \\
$r$ ($\ket{\psi_0} = \ket{1}\ket{0}^{\otimes n-1}$) & 2 & 4 & 6 & 8 & 10 & 12 \\
\bottomrule
\end{tabular}
\caption{Degree of model reduction possible for the collective rotation Hamiltonian model, \erf{eq:eg_coll_spin} as a function of number of spins. $|\mathscr{B}(\Ham)|$ is the size of the Burnside basis and $r$ is the dimension of the reduced subspace under the specified initial state. \label{tab:coll_spin}}
\end{table*}

\end{example}

\begin{example}[Transverse-field Ising model]
A paradigmatic spin chain model exhibiting many important many-body phenomena is the transverse field Ising Hamiltonian \cite{Sac-1998}:
\begin{equation}
\label{eq:eg_ising}
		\Ham_{\textrm{Ising}}(\lambda) = -B \sum_{j=1}^n \sigma_x^{(j)} - J \left(\sum_{j=1}^{n-1} \sigma_z^{(j)}\sigma_z^{(j+1)} +  \sigma_z^{(n)}\sigma_z^{(1)} \right),
\end{equation}
where the two model parameters are $\lambda=(B,J)$, and we assume periodic boundary conditions so that the spins at the ends of the chain are coupled.

In this case, we numerically calculated the reduced subspace under two possible initial conditions, and the results are shown in table \ref{tab:ising}. The first initial state is the completely polarized state (all spins aligned in the direction of the transverse field) $\ket{\psi_0}=\ket{+}^{\otimes n}$, where $\ket{+} = \frac{1}{\sqrt{2}}(\ket{0}+\ket{1})$, and the second is the ground state of \erf{eq:eg_ising} for $B=0.05, J=1$, notated by $\ket{\texttt{gs}(0.05,1)}$. The second initial state is motivated by quench dynamics experiments, where a many-body system is prepared in the ground state of a model at some parameter values $\lambda_0$, and then the model parameters are quickly changed (quenched) to some other values $\lambda_1$. The resulting dynamics can be interesting and in many cases informative about the equilibrium phase diagram of the model, \eg \cite{Dziarmaga:2006ue,Heyl:2013uc}. Physically implementing quenched dynamics is becoming increasingly feasible, \eg \cite{Edwards:2010iv}, and therefore predictive simulations of such dynamics are extremely valuable. Although the particular model in \erf{eq:eg_ising} is exactly solvable, it serves to illustrate the compressibility of such many-body models, and the dimensions in table \ref{tab:ising} indicate how much model reduction is feasible for such dynamical simulations.

To explicitly demonstrate model reduction for this model, in Fig. \ref{fig:ising} we show the net transverse magnetization for an $8$-spin transverse field Ising chain with $\ket{\psi_0}=\ket{\texttt{gs}(0.05,1)}$ quenched to various various parameter values, simulated using the full order model ($d=2^8$) and its reduced order version generated by projecting onto the reduced subspace for the model ($r=24$). The reduced subspace was constructed via the Burnside basis method of \textsection  \ref{sec:burnside_alg} and Algorithm \ref{alg_burnside}. The reduced order dynamical model is given by the projected dynamics in \erf{eq:red_order_model}, and the reduced-order observable is generated by projection: $\hat{O} = \Phi\dg O \Phi$. We see complete agreement between the two simulations for any value of $\lambda$, with errors on the order of numerical precision. To confirm that this agreement is not due to this particular observable being robust to Hilbert space truncation, we also show in Fig. \ref{fig:ising} that if the invariant subspace is truncated by removing just one basis vector, the reduced order model can disagree with the full order model.

\begin{table*}\centering
\ra{1.3}
\begin{tabular}{@{}lrrrrrrr@{}}
\toprule
$n$ & 2 & 3 & 4 & 5 & 6 & 7 & 8\\
\midrule
$|\mathscr{B}(\Ham)|$ & 6 & 10 & 27 & 50 & 126 & 250 & 536\\
$r$ ($\ket{\psi_0} = \ket{+}^{\otimes n}$) & 2 & 2 & 4 & 4 & 8 & 8 & 16 \\
$r$ ($\ket{\psi_0} = \ket{\texttt{gs}(0.05,1)}$) & 3 & 4 & 6 & 8 & 12 & 16 & 24 \\
\bottomrule
\end{tabular}
\caption{Degree of model reduction possible for the transverse field Ising model, \erf{eq:eg_ising}, as a function of number of spins. $|\mathscr{B}(\Ham)|$ is the size of the Burnside basis and $r$ is the dimension of the reduced subspace under the specified initial state. \label{tab:ising}}
\end{table*}

\begin{figure}[h!]
\centering
\subfigure[{\small Dynamics projected onto full reduced subspace, $r=24$.}]{\includegraphics[scale=0.25]{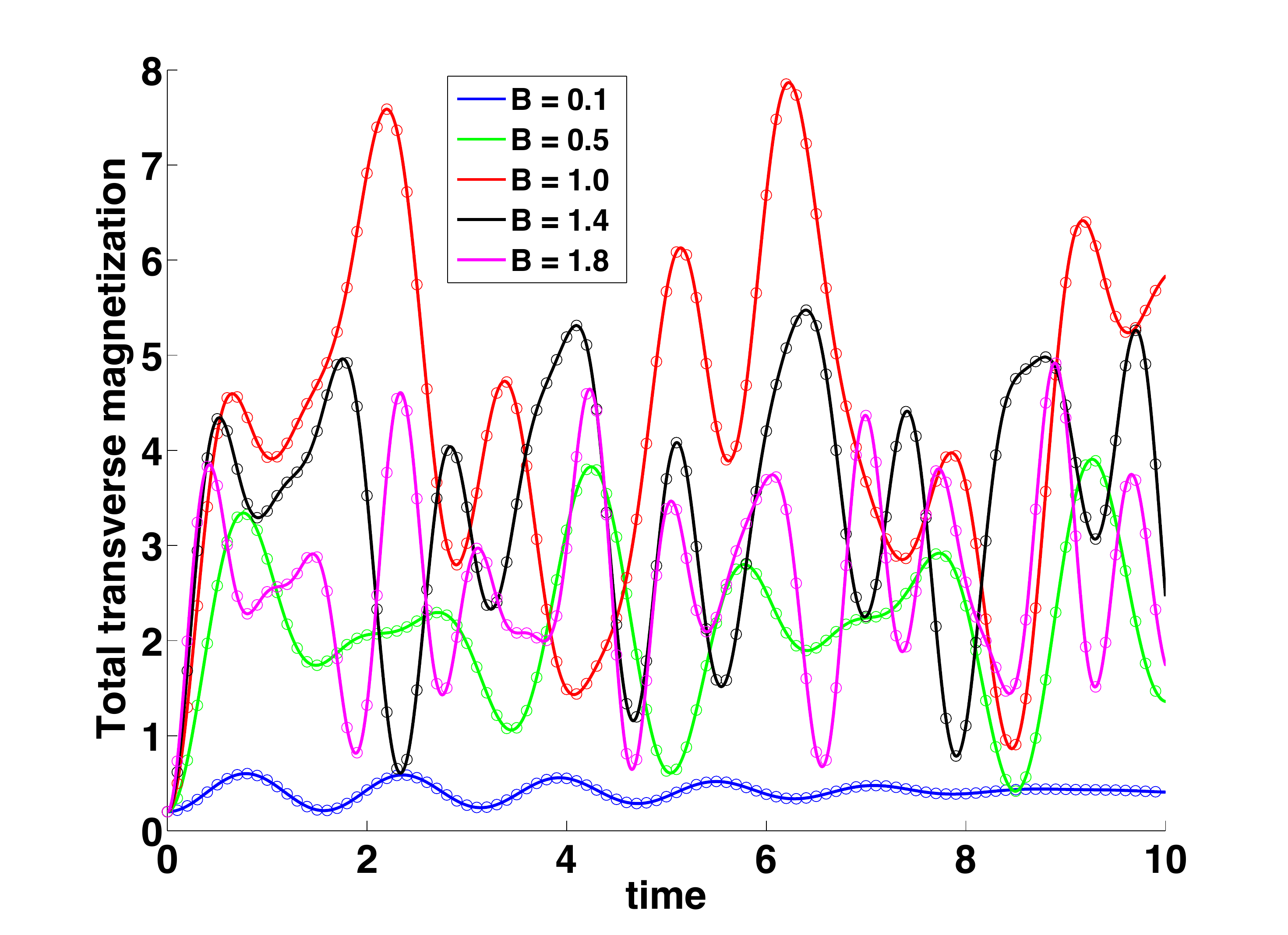}}
\subfigure[{\small Dynamics projected onto truncated reduced subspace, $r=23$.}]{\includegraphics[scale=0.25]{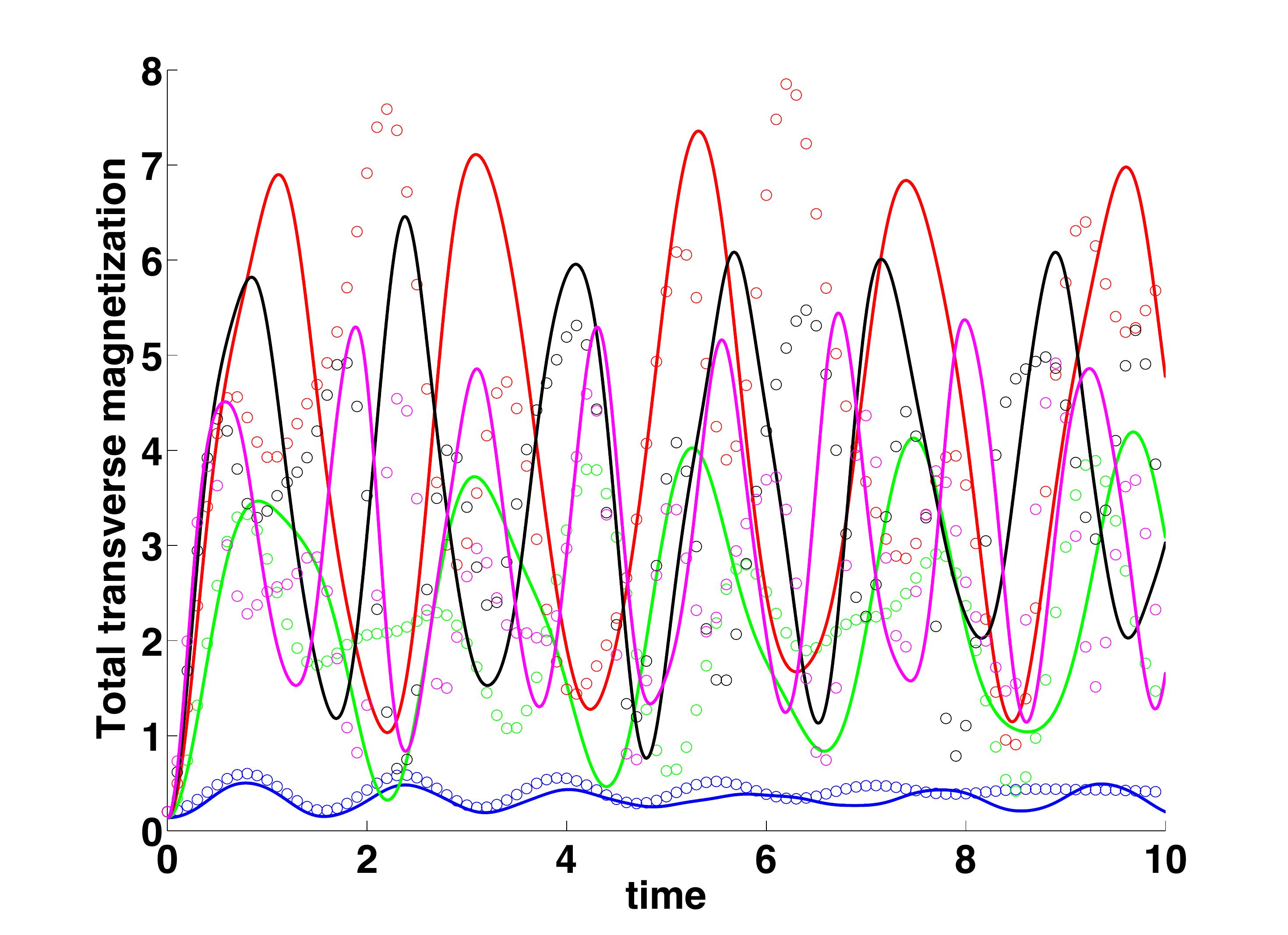}}
\caption{\small \label{fig:ising} Net transverse magnetization (the observable $\expect{\sum_{j=1}^8 \sigma_x^{(j)}}$) for an $8$-spin transverse-field Ising chain under quenched dynamics. The colors represent dynamics resulting from different quenching parameters ($J=1$ and $B$ specified in the legend), all starting from the initial state $\ket{\psi_0} = \ket{\texttt{gs}(0.05,1)}$. The curves correspond to dynamics generated by the reduced order model whereas the circles represent dynamics generated by the full order model. (a) shows complete agreement between the full order model and the reduced order model defined by the invariant subspace, while (b) shows disagreement when the dynamics are projected onto a subspace of the reduced order model, of codimension 1.}
\end{figure}

\end{example}

\begin{example}[Random transverse-field Ising model]
A more complex spin chain Hamiltonian is the random transverse-field Ising model \cite{Fisher:1995gb}:
\[
		\Ham_{\textrm{random Ising}}(\lambda) = \sum_{j=1}^n B_j \sigma_x^{(j)} +  \sum_{j=1}^{n-1} J_j \sigma_z^{(j)}\sigma_z^{(j+1)},
\]
where the $2n-1$ model parameters are $\lambda=(\{B_j\}_{j=1}^n , \{J_j\}_{j=1}^{n-1})$, and we assume open boundary conditions (\ie the spins on the ends of the chain are only coupled to one neighbor). This is a pure Pauli Hamiltonian with $|\Sb_\lambda(\Ham)| < 2n$ and therefore by Theorem \ref{pauli:thm1} we know it possesses a non-trivial invariant subspace. 

In computing the reduced subspace for this model we can exploit the simplifications enabled by the binary vector representation of $\Pb_n$ as detailed in \textsection  \ref{sec:pauli}. Computing the Burnside basis using this representation and Algorithm \ref{alg:pauli_burnside_basis_2} in Appendix \ref{app:burnside_alg}, we find that $|\mathscr{B}(\Ham)| = \frac{4^n}{2}$. Now, consider an initial state where all spins are aligned in the direction of the transverse field:  $\ket{\psi_0}=\ket{+}^{\otimes n}$. In this case, the inner product $\expect{\bp}$ for $\bp \in \Pb_n$ specified in the binary vector representation is easily computable since the initial state is a tensor product of $\sigma_x$ eigenstates. Using this fact, we can readily compute the maximal linearly independent set of vectors in $\{B_j \ket{\psi_0}\}_{j=1}^{|\mathscr{B}(\Ham)|}$ using the Gramian matrix method detailed in \textsection  \ref{sec:construct_pauli}, and we find that the reduced order model dimension is $r = 2^{n-1}$, half of the formal Hilbert space dimension. 
\end{example}

\section{Discussion}
\label{sec:disc}
In this work we have initiated the study of model reduction for many-body quantum systems via identification and exploitation of invariant subspaces that host all dynamics generated by a parameterized Hamiltonian from a given initial state. The methods developed here can be applied to simplify simulations of quantum many-body systems for many applications, including quenched quantum dynamics, adiabatic quantum evolution, and quantum control. The degree of efficiency improvement afforded by the model reduction methods developed here are dependent on the exact model being considered. In the worst case, where no invariant subspace exists for the dynamics, no improvement is possible. However, as the examples presented above demonstrate, in many cases significant efficiency improvements are possible.

The methods developed in this work motivate several possible directions for future research, including: (1) constructing more efficient certificates for determining the (im)possibility of a model reduction than the direct computation of a basis for $\A(\Ham)$; (2) developing time and space saving techniques for attaining a maximal linearly independent subset of a given set of vectors, or for computing its size, in application to constructing a basis for $\A(\Ham)$ or determining $\dim(\A(\Ham))$, respectively; and (3) extending the time sampling methods to more general classes of dynamical evolutions and providing sampling strategies in terms of the parameters $\lambda$ in order to exhaust the full subspace $\A(\Ham)\cdot \ket{\psi_0}$. We elaborate upon each of these directions in the following.

Firstly, concerning the results of \textsection  \ref{sec:cert_general}, the key benefit to the current certificate for the existence of non-trivial proper invariant subspaces is that computing the Burnside basis is also essential for generating the reduced subspace for the model: $\A(\Ham)\cdot \ket{\psi_0}$. Therefore any method to increase the efficiency of implementing Algorithm \ref{alg_burnside} in Appendix \ref{app:burnside_alg} would be extremely beneficial. On the other hand, if only a sufficient condition for the \emph{possibility} of a model reduction is required, it is possible that other methods for computing $\dim(\A(\Ham))$ or deriving bounds on it may be more efficient. Conversely, there also exist certificates for determining \emph{impossibility} of a model reduction, for example, following Laffey \cite[\textsection~5]{laffey1986simultaneous}, the model $\Ham$ has \emph{no} reduced subspace if the characteristic polynomial $\det(x_0 H_0 + \cdots  + x_M H_M + x_{M+1} (\ket{\psi_0}\bra{\psi_0}) - x I_d)$ in the $M+3$ complex variables $x, x_j$, with $\{H_j\}_{j=0}^M = \Coeff(\Ham)$, does \emph{not} split into linear factors. This can practically be computed by absolute factorization methods or through comparing the homogenization of the polynomial against the corresponding Chow variety \cite[Ch. 4]{gelfand2008discriminants}. Efficient implementations and more certificates in these directions require further attention.

Another crucial computational element in forming a reduced order model, via the Burnside basis or time sampling, is finding a maximal linearly independent set from a collection of vectors; \eg the collection $\{B_j \ket{\psi_0}\}_{j=1}^{|\mathscr{B}(\Ham)|}$ in the construction via the Burnside basis. These vectors are of length $d$ and therefore grow exponentially with the number of elementary degrees of freedom. This makes explicit computation of the maximal linearly independent set infeasible for large models. In \textsection  \ref{sec:pauli} we exploited Pauli group structure and assumptions on the initial state to perform this computation using a Gramian matrix which avoids explicitly working with the $d$-dimensional vectors. Other techniques to extract the linearly independent set efficiently will be useful in reducing the burden of computing the reduced order model. 

On a related note, the only special class of Hamiltonians that we have examined in this work are the Pauli Hamiltonians for many-body spin-$1/2$ models. It might be fruitful to study if other special cases, \eg commuting Hamiltonians where invariant subspaces are equivalent to common eigenspaces for all terms in the Hamiltonian, possess structure that makes identification and construction of reduced order models particularly easy.

The results established in \textsection  \ref{sec:sample} quantifying the effectiveness of forming the model reduction matrix $\Phi$ by sampling the time evolution of the full order model provide another avenue for future work. The techniques used to develop these results are very general and amenable to extension beyond quantum dynamics generated by unitary evolution. Using these techniques, it should be possible to establish conditions that specify when time sampling captures all dynamical modes for dynamics generated by any map analytic in $t$ and $\lambda$. Such a result, which would formulate conditions for the probabilistic abundance of $\emph{good}$ time samples for any analytic map, is relevant to many empirical model reduction techniques that rely on sampling to construct reduced order models \cite{Antoulas:2009tb,Schilders:2008uc}. 

Although we have established conditions under which time sampling is effective for constructing the model reduction matrix $\Phi$, it remains an open question as to how to best perform the sampling in time and parameters $\lambda$. It was remarked in \textsection  \ref{sec:sample} that by collecting snapshots of states generated by dynamics at multiple values of the parameters, $\lambda_1, \lambda_2, ...$, one increases the probability of generating all of reduced subspace $\A(\Ham)\cdot \ket{\psi_0}$. A promising direction for future work is constructing strategies for performing this sampling in a systematic manner.

Finally, an important practical issue reserved for future study is the numerical stability of the algorithms developed in this work. For example, there are several points in implementing Algorithm \ref{alg_burnside} in Appendix \ref{app:burnside_alg} or determining a linearly independent set from a collection of vectors where determinants or ranks must be computed, and it would be interesting to study how the numerical stability of these computations affects the reduced order models that are constructed.

Progress on any of the fronts outlined above will contribute to the further development of reduced order modeling techniques for quantum mechanical systems, and hence reduce the computational burden of simulating such systems.

\section*{Acknowledgements}
A.K. would like to thank Martin Harrison and Ryan Mohr for many helpful discussions pertaining to the topics of invariant subspaces and the structure of quantum dynamical evolution. M.S. thanks Kevin Carlberg (SNL-CA) for many helpful discussions on the topic of model reduction.
This work was supported by the Laboratory Directed Research and Development program at Sandia National Laboratories. Sandia is a multi-program laboratory managed and operated  by Sandia Corporation, a wholly owned subsidiary of Lockheed Martin Corporation, for the United States Department of Energy's National Nuclear Security Administration under contract DE-AC04-94AL85000.

\bibliographystyle{ieeetr}
\bibliography{main}

\clearpage

\appendix
\section{Burnside's Theorem on Matrix Algebras}
\label{app:burnside_proof}
We begin with the following classical theorem of Burnside, which has become a core part of the theory of
group representations:
\begin{theorem}[Burnside \cite{burnside1905condition}] \label{burnside_orig}
	Let $\Gr$ be a group and $\pi:\Gr\rightarrow \GL_d(\CC)$ a representation. Then, $\pi$ is irreducible if and only if $\im(\pi)$ spans $\Mat_d(\CC)$. \qed
\end{theorem}

From this familiar version of Burnside's theorem, we can derive Theorem \ref{red:burnside}, which is a version applicable to matrix subalgebras that we have used throughout this work. Although many proofs of Theorem \ref{red:burnside} exist in literature, nearly all of them are independent of Burnside's original result and the connection has become somewhat a part of folklore. In order to make the connection apparent, we produce a proof below that uses directly the above theorem of Burnside on group representations, making the theorem for matrix subalgebras simply a:

\begin{corollary}[of Burnside's theorem]
	A subalgebra $\A \leq  \Mat_d(\CC)$ is irreducible if and only if $\A = \Mat_d(\CC)$.
\end{corollary}

\begin{proof}
Given any subalgebra $\A \leq  \Mat_d(\CC)$, let $\Gr := \langle \A \cap  \GL_d(\CC)\rangle $ be the subgroup of $\GL_d(\CC)$ generated by the invertible matrices in $\A$. The strategy for the proof will be to show $\Gr \subseteq  \A$ and that subspace invariance under $\Gr$ is equivalent to subspace invariance under $\A$: $\Inv(\Gr) =\Inv(\A)$. Then, an application of Theorem \ref{burnside_orig}, by letting the representation $\pi:\Gr \hookrightarrow  \GL_d(\CC)$ be the canonical inclusion map, will show the desired result.

To show that $\Gr$ is a subset of $\A$, note that $\A$ contains the inverses of its invertible elements. This is easy to see by writing the inverse of $A \in  \A$ as a polynomial in $A$ per the Cayley-Hamilton theorem: $A^{-1} = \frac{(-1)^{d-1}}{\det(A)}(A^{d-1}+c_{d-1}A^{d-2} + ... + c_1 I_d)$, for scalars $c_k$. Then by the closure of $\A$ under products and linear combination we have that $\Gr \subseteq  \A$.

Next, we show that $\Inv(\Gr) =\Inv(\A)$. First, the direction $\Inv(\A) \subseteq \Inv(\Gr)$ is easily seen from $\Gr \subseteq  \A$. The reverse direction is less trivial since $\A$ has singular operators (\viz $\A$ properly contains $\Gr$), but we will overcome this difficulty by a limiting argument. Consider the polynomial $p_\tau(x) = (1 - \tau)x + \tau$ and note that for any singular $A \in  \A$, the polynomial $\det(p_\tau(A))$ is not identically zero (as a function of $\tau$) and has a root at $\tau=0$. So it must be non-zero in a punctured neighborhood about $\tau=0$, and hence there is a punctured neighborhood $\Nbd_A \subseteq  \GL_d(\CC)$ about $A$. We collect these neighborhoods in $\Nbd := \cup _{A \in  \A \wo \GL_d(\CC)} \Nbd_A$ and note that $\Nbd \subset  \A \cap  \GL_d(\CC)$. Now, since vector subspaces are topologically closed 
\cite[\textsection~11.294]{Glazman:2006}
, the invariance of subspaces is preserved under taking pointwise limits from $\Nbd$ converging to singular matrices in $\A$; in other words, if $\cl{\Gr}$ is the topological closure 
\cite[\textsection~17]{munkres2000topology} 
of $\Gr$ and $V \in  \Inv(\Gr)$, then $\A = \cl{\Gr}$ and $V \in  \Inv(\cl{\Gr}) = \Inv(\A)$, whence we see that $\Inv(\Gr) \subseteq  \Inv(\A)$.

Having shown $\Inv(\Gr) =\Inv(\A)$, we can proceed with proving the corollary. First, if $\A = \Mat_d(\CC)$, then $\Gr = \GL_d(\CC)$ by definition and $\vspan(\Gr)= \Mat_d(\CC)$. Then by Theorem \ref{burnside_orig} $\Gr$ has no proper invariant subspaces, and $\Inv(\Gr) =\Inv(\A)$ implies the same for $\A$, meaning that $\A$ is irreducible. To go the other way, assume that $\A$ is irreducible, then by $\Inv(\Gr) =\Inv(\A)$ so is $\Gr$. By Theorem \ref{burnside_orig} this means that $\vspan(\Gr)= \Mat_d(\CC)$, but because $\Gr \subseteq  \A$, if $\vspan(\Gr)= \Mat_d(\CC)$, then $\A= \Mat_d(\CC)$.
\end{proof}

\begin{remark}
	See \cite{curtis1999pioneers} for more of a historical discussion on the development of this theorem and
	the connection with Frobenius and Schur.
\end{remark}

\section{Explicit algorithms for calculating the algebraic Burnside basis}
\label{app:burnside_alg}
Algorithm \ref{alg_burnside} details the procedure to generate the Burnside basis for a general Hamiltonian $\Ham$, which is a maximal linearly independent subset (over $\CC$) of the monoid generated by taking products of the operators in $\Coeff(\Ham)$. In terms of Definition \ref{def:algebra}, $\A(\Ham)$ is the linear span of monomials in $\Coeff(\Ham)$. The crux of the algorithm resides in the fact that if $\cL_k$ is the collection (called a \emph{layer}) of all monomials of degree at most $k$, then we have a chain of inclusions,
\[
	\vspan{\cL_1} \subseteq  \vspan{\cL_2} \subseteq  \cdots  \subseteq  \vspan{\cL_k} \subseteq  \cdots ,
\]
where if at some point $k$ in the chain of inclusions equality holds, so that $\vspan{\cL_k} = \vspan{\cL_{k+1}}$, then it must be that for all $m \geq  k$, $\vspan{\cL_m} = \vspan{\cL_k}$ and moreover, $\vspan{\cL_k} = \A(\Ham)$ (see \cite[\textsection~4]{laffey1986simultaneous}). One method of achieving this (step \ref{alg_burnside:main_loop}.(\ref{alg_burnside:spancheck})) is by adding only linearly independent monomials of degree $k$ to $\cL_{k-1}$, at the $k$-th iteration. Then, the above inclusion chain also holds for this collection of linearly independent layers, the stopping criterion simply becomes that $\cL_k = \cL_{k+1}$ and clearly $\cL_k = \Burn(\Ham)$. We detail a second possible implementation that may be useful for scaling upwards, by off-loading the linear independency check at each step of the iteration to a post-processing step.

Algorithm \ref{alg:pauli_burnside} is a straightforward modification of Algorithm \ref{alg_burnside} that exploits additional structure in pure Pauli Hamiltonians of the form:
\begin{equation}
\label{eq:app_pauli}
\Ham(\lambda) = \sum_{k=1}^M \lambda_k \bp_k
\end{equation}
where each $\bp_k$ is an $n$-spin generalized Pauli operator. In this case, $\Coeff(\Ham) = \{\bp_k\}_{k=1}^M$, and properties of the elements of $\Pb_n$ described in the main text simplify operations within the Burnside basis algorithm. Recall that $\np{X}$ is the binary vector representation of $X \in \Pb_n$.

Finally, Algorithm \ref{alg:pauli_burnside_basis_2} is a different method of generating the Burnside basis for a pure Pauli Hamiltonian that exploits all the structure of $\Pb_n$ and its binary vector representation. It is considerably more efficient than Algorithm \ref{alg:pauli_burnside} because it exploits the fact that $\Pb_n$ is abelian to reduce the number of binary additions (equivalent to multiplications of elements of $\Pb_n$) required. To do so, we use the concept of a Gray code to cycle through the binary additions necessary to generate the basis while avoiding repetitions that would result if we ignored the abelian nature of $\Pb_n$. This is a method for quickly generating subspaces of the binary vector space as used in the Method of Four Russians; the correctness of the algorithm is argued for in \cite[\textsection~9.2]{bard2009algebraic}.

Note that the return value for the general Algorithm \ref{alg_burnside} is a set of operators while the return value for the two modified algorithms (for pure Hamiltonians) is a set of binary vectors, each of which corresponds to an element of $\Pb_n$.

In the following XOR (entry-wise addition modulo $2$) is denoted by $\oplus$.

\clearpage

\IncMargin{1em}
\begin{procedure}[H]
	\SetKwData{Mons}{Monomials}\SetKwData{NewMons}{NewMonomials}
	\SetKwInOut{Input}{input}\SetKwInOut{Output}{output}
	\Indm
	\Input{\Mons}
	\Output{\NewMons}
	\Indp
	$\NewMons \leftarrow \varnothing $	\\
	\For{$W \in  \Mons, X \in  \Coeff(\Ham)$}{
		$\NewMons \leftarrow \NewMons \cup  \{W \cdot  X\}$	\\
	}
	\caption{incMonomials(): Takes a set of monomials and constructs all monomials of degree one greater, from the generating operators $\Coeff(\Ham)$.}
\end{procedure}

\vspace{2em}

{\RestyleAlgo{algoruled}
\begin{algorithm}[H]
	\SetAlgoNoLine
	\SetKwData{Layer}{Layer}\SetKwData{Mons}{Monomials}\SetKwData{NewMons}{NewMonomials}
	\SetKwFunction{Span}{span}
	\SetKwInOut{Input}{input}\SetKwInOut{Output}{output}
	
	\Indm
	\Input{$\Coeff(\Ham)$}
	\Output{$\Burn(\Ham)$}
	\Indp
	
	$\Layer \leftarrow \varnothing , ~ \Mons \leftarrow \{ I_d \}$	\\
		
	Iteratively run $\Mons \leftarrow \incMonomials{\Mons}$ and $\Layer \leftarrow \Layer \cup  \Mons$ until the stopping criterion
	is met that \Span{\Layer} is unchanged for two consecutive iterations. This check can be implemented
	in a variety of ways; two of the possible implementations are:
	\begin{enumerate}
	\item	Modify the loop of \incMonomials{} to check that each new monomial is not in
				\Span{$\Layer \cup  \Mons \cup  \NewMons$}
			before adding it to \NewMons. This check can be implemented in a variety of ways, one of
			which would be to generate a new matrix from the vectorizations of all of the matrices in
			question and compare ranks. The stopping criterion is then simply that the new collection of
			monomials is not empty. In this regime, after the halt, clearly $\Layer = \Burn(\Ham)$.
			\label{alg_burnside:spancheck}
	\item	Na\"ively, one only needs at most $d^2$ iterations before the stopping criterion must be met. This
			upper bound can be improved due to a result of Paz \cite{paz1984application} to
				$k := \lceil \frac{d^2 + 2}{3}\rceil $.
			Thus, one can avoid the check of spans at each step and simply form all $k$ layers, then post-process
			\Layer to find a maximal linearly independent subset to serve as the Burnside Basis $\Burn(\Ham)$. See
			\cite[\textsection~4]{laffey1986simultaneous} for more discussion (including that of a conjecture of Paz that
			$k = 2d - 2$, which is shown there to be a \emph{lower} bound in general).
			\label{alg_burnside:paz}
	\end{enumerate}
	\label{alg_burnside:main_loop}

	\caption{Generating the Burnside basis for a general Hamiltonian $\Ham$.}\label{alg_burnside}
\end{algorithm}
}

\clearpage

\begin{procedure}[H]
	\SetKwData{Vectors}{Vectors}\SetKwData{NewVectors}{NewVectors}
	\SetKwInOut{Input}{input}\SetKwInOut{Output}{output}
	\SetKwFunction{Phi}{$\varphi$}
	\Indm
	\Input{\Vectors}
	\Output{\NewVectors}
	\Indp
	$\NewVectors \leftarrow \varnothing $	\\
	\For{$W \in  \Vectors, X \in  \Coeff(\Ham)$}{
		$\NewVectors \leftarrow \NewVectors \cup  \{W \oplus  \Phi{$X$}\}$	\\
	}
	\caption{incVectors(): Version of \incMonomials{} for Pauli Hamiltonians.}
\end{procedure}

\vspace{2em}

{\RestyleAlgo{algoruled}
\begin{algorithm}[H]
	\SetAlgoNoLine
	\SetKwData{Layer}{Layer}\SetKwData{Vectors}{Vectors}\SetKwData{NewVectors}{NewVectors}
	\SetKwFunction{Span}{span}\SetKwFunction{Phi}{$\varphi$}
	\SetKwInOut{Input}{input}\SetKwInOut{Output}{output}
	
	\Indm
	\Input{$\Coeff(\Ham) = \Sb_\lambda(\Ham)$}
	\Output{$\varphi[\Burn(\Ham)]$}
	\Indp
	
	$\Layer \leftarrow \varnothing , ~ \Vectors \leftarrow \{ \vec{0} \}$	\\
	
	Iteratively run $\Vectors \leftarrow \incVectors{\Vectors}$ and $\Layer \leftarrow \Layer \cup  \Vectors$ until the stopping criterion
	is met that $\Layer$ is unchanged for two consecutive iterations. This check can be implemented in a variety
	of ways; two of the possible implementations are:
	\begin{enumerate}
	\item	Modify the loop of \incVectors{} to check that each new binary vector $A$ is not in
				$\cL := \Layer \cup  \Vectors \cup  \NewVectors$
			before adding it to \NewVectors. One possible implementation of this check is through binary AND
			operations; \ie $A \in  \cL$ iff $(\exists )C \in  \cL$ such that $A ~ \texttt{AND} ~ C = 0$. The stopping criterion
			then becomes that the new collection of binary vectors is not empty and finally, $\Layer = \varphi[\Burn(\Ham)]$.
	\item	The alternate method of step \ref{alg_burnside:main_loop}.(\ref{alg_burnside:paz}) of Algorithm
			\ref{alg_burnside} can again be applied here, except that the post-processing step need only
			eliminate any duplicate binary vector entries.
	\end{enumerate}

	\caption{Na\"ive speed-up of Algorithm \ref{alg_burnside} for generating the Burnside basis for a pure Pauli Hamiltonian of the form \erf{eq:app_pauli}.}
	\label{alg:pauli_burnside}
\end{algorithm}
}

\vspace{5em}

\begin{algorithm}[H]
	\SetKwData{Layer}{Layer}\SetKwData{Gray}{Gray}\SetKwData{Flip}{flip}
	\SetKwFunction{Log}{$\log_2$}\SetKwFunction{Phi}{$\varphi$}
	\SetKwInOut{Input}{input}\SetKwInOut{Output}{output}
	\Input{$\Coeff(\Ham) = \Sb_\lambda(\Ham)$}
	\Output{$\mathcal{V} = \Phi{$\Burn(\Ham)$}$}
	\Begin{
		$\mathcal{V} \leftarrow \varnothing $	\\
		$\ell \leftarrow \st\Coeff(\Ham)\st$	\tcp*[r]{Order elements of $\Coeff(\Ham)$ from $1 \ldots  \ell$.} 
		\Gray $\leftarrow$ Gray code for $\ell$ bits	\\
		$A \leftarrow 0$	\\
		\For{$j \leftarrow 2$ \KwTo $2^{\ell}$}{
			$\Flip \leftarrow \Log{$\Gray[j-1] \oplus  \Gray[j]$}$ \label{grayflip}	\\
			$X \leftarrow X_{\Flip + 1} \in  \Coeff(\Ham)$	\\
			$A \leftarrow A \oplus  \Phi{$X$}$	\\
			\If{$A \not\in  \mathcal{V}$}{
				Append $A$ to $\mathcal{V}$	\\
			}
		}
	}
	\caption{Efficient algorithm for generating the Burnside basis for a pure Pauli Hamiltonian of the form \erf{eq:app_pauli}. The idea is to generate all linear combinations of $\varphi(\Coeff(\Ham))$ according to a Gray codes table, which avoids the layering method needed for less structured sets of operators. The purpose of line \ref{grayflip} is to calculate the bit that is flipped between two neighboring elements of the Gray code.}
	\label{alg:pauli_burnside_basis_2}
\end{algorithm}
\DecMargin{1em}

\end{document}